\documentclass[11pt,a4paper,reqno]{amsart}
\usepackage{fullpage} 
\usepackage{amsfonts}
\usepackage{amsthm}
\usepackage[utf8]{inputenc}
\usepackage[T1]{fontenc}
\usepackage{amsmath} 
\usepackage{graphicx}
\usepackage{enumitem}
\usepackage{wrapfig}
\usepackage{dsfont}
\usepackage{braket}
\usepackage{comment}
\usepackage{amsxtra}
\usepackage{hyperref}
\usepackage[capitalise]{cleveref}
\usepackage{xkvltxp}
\usepackage{fixme}

\newtheorem{de}{Definition}
\newtheorem{theo}[de]{Theorem}    
\newtheorem{prop}[de]{Proposition}
\newtheorem{lem}[de]{Lemma}
\newtheorem{rem}[de]{Remark}
\newtheorem{cor}[de]{Corollary}

\DeclareMathOperator{\tr}{Tr}

\DeclareMathOperator{\loc}{loc}

\newcommand{\nn}{\nonumber}
\newcommand{\dip}{\textrm{dip}}

\newcommand{\dd}{\mathrm{d}}

\title{Derivation of the time-dependent Gross-Pitaevskii equation for dipolar gases}
\date{}
\author{Arnaud Triay}

\begin{document}
 
\begin{abstract}
	We derive the time-dependent dipolar Gross-Pitaevskii (GP) equation from the $N$--body Schrödinger equation. More precisely we show a norm approximation for the solution of the many body equation as well as the convergence of its one-body reduced density matrix towards the orthogonal projector onto the solution of the dipolar GP equation. We consider the interpolation regime where interaction potential is scaled like $N^{3\beta-1}w(N^{\beta} (x-y))$, the range of validity of $\beta$ depends on the stability of the ground state problem. In particular we can prove the convergence on the one-body density matrix assuming $\widehat{w} \geq 0$ and $\beta < 3/8$.
\end{abstract}

\maketitle
\setcounter{tocdepth}{1}
\tableofcontents
\section{Introduction}

The phenomenon of Bose-Einstein Condensation (BEC) predicted in 1924 \cite{Bose-24,Ein-24} and experimentally observed in 1995 \cite{CorWie-95}, has triggered a vast interest in the mathematical physics community. Initially analyzed for a gas of ideal particles, it has been a mathematical challenge to prove its persistence in the presence of interactions. Most of the first studies on BEC focused on the repulsive case, where the interaction is assumed to be non-negative or, at least, with a non-negative scattering length. Alongside, as experimental physicists mastered the creation process of condensates, it has been made possible to condense chemical elements with more complex interactions. In particular, the realization of dipolar BEC was achieved in 2005 for Chromium by Griesmaier et al \cite{GriWerHenStuPfa-05} and is still an active domain of research \cite{Beaufils-08,LuBurYouLev-11,Aikawa-12}. 

The nature of the dipolar interaction opened the way to a great variety of new properties such as a stable/unstable regime, the roton-maxon shape of the excitation spectrum \cite{SanShlLew-03,HenNatPoh-10} or the existence of a droplet state \cite{Chomaz-16,Baillie-16}. See \cite{LahMenSanLewPfau-09} for a survey. Yet, the dipolar interaction rarely fits the framework of the standard mathematical analysis of Bose-Einstein condensation as it is long-range and partially attractive.

In this work we give the first the derivation of the dipolar Gross-Pitaevskii (GP) equation from the $N$-body Scrhödinger dynamics. The GP equation determines the evolution of the common wave function of all the particles in the condensate. In the $3$D case and in the Gross-Pitaevskii regime, the rigorous derivation for a repulsive interaction was proven in a series of works by Erdös, Schlein and Yau \cite{ErdSchYau-06,ErdSchYau-07,ErdSchYau-09,ErdSchYau-10}.
%\cite{BenOliSch-15,BreSch-17,BocCenSch-17,KnoPic-10,NamNap-17,NamNap-17b,NamNap-17c,Pickl-11,Pickl-15,JebLeoPic-16,Golse-13,CheHaiPavSei-14,CheHaiPavSei-15,BreNamNapSch-17,LewNamSch-15}. 
This is the regime where the scattering length is of order $N^{-1}$, it corresponds in the $N$ body setting to a scaled interaction potential of the form $N^{2}w(N \cdot)$ for some fixed $w\in L^1(\mathbb{R}^3)$. In this regime, the scattering process plays an important role and has to be precisely taken into account. Like many other works in the subject, we will focus on an interpolation regime between the Gross Pitaevskii and the Hartree regime where the potential is scaled like $N^{3\beta-1} w(N^\beta \cdot)$  with $0<\beta < 1$. This case is easier because no significant correlation structure is expected to take place. In our setting, the difficulty comes from the attractive part of the interaction potential that may cause instabilities, as we will discuss.
%\cite{RodSch-09}

Removing the non-negativity assumption on the interacting potential or studying the attractive case is a difficult task since the system may not be stable of the second kind. For this reason, the focusing case $w\leq 0$ is only globally well-posed in low dimensions ($d\leq 2$) \cite{NamNap-17d}. The derivation of the $1$D and $2$D focusing cases was provided by Chen and Holmer \cite{CheHol-16,CheHol-16b} using the BBGKY hierarchy method. It relies on compactness arguments and therefore does not give any information on the rate of convergence. Later Jeblick and Pickl \cite{JebPic-18} using a method of Pickl \cite{Pickl-11} gave another proof in the 2D case yielding a precise estimate on the rate of convergence in trace norm for the density matrices. Then Nam and Napiórkowski  \cite{NamNap-17d} obtained the norm approximation of the $N$-body wave function. This result gives the fluctuations around the condensate and implies the convergence in trace class of the density matrices. In these works, the range of $\beta$ depends on the stability of the second kind analyzed in \cite{LewNamRou-15,LewNamRou-17b}.

For the $3$D case, Pickl \cite{Pickl-10} could deal with $0<\beta< 1/6$ assuming the interaction $w$ is compactly supported, spherically symmetric and bounded. In \cite{Chong-16}, Chong gave, under the same assumptions and additionally that $w \leq 0$, another proof of the convergence of the density matrices. Later, Jeblick and Pickl \cite{JebPic-18b} proved the convergence of the density matrices in the GP regime ($\beta = 1$) for a class of non-purely non positive potentials, namely, for which one has stability of the second kind. The class of potentials treated in this last work is quite specific and does not include long-range interaction of the type considered here.

In this paper, we show, in the case of long-range interactions, the norm approximation in $L^{2}(\mathbb{R}^{3N})$ of the solution of the $N$-body Schrödinger equation by the solution of the Bogoliubov evolution as well as the convergence in trace norm of one-body reduced density matrix towards the orthogonal projector onto the solution of the dipolar GP equation. In particular, we remove the non-negativity assumption on the interaction potential and we are able to consider the dipole-dipole interaction (DDI), given by 
\begin{equation*}
K_{\dip}(x) = \frac{1-3\cos (\theta_x)^2}{|x|^3},
\end{equation*}
where $\theta_x$ is the angle between $x$ and a fixed direction along which the dipoles are aligned. The exact type of potentials we consider will be detailed later. The derivation of the Gross-Pitaevskii energy for the ground state of a Bose gas with dipolar interaction was studied in \cite{Tri-18}.

%[Time independent] \cite{LieSei-02,LieSei-09,LieYng-01,LieSeiYng-00}
%\cite{GuoSei-14} (GS attractive)
A quantitative method developed in \cite{Pickl-11} consists in applying the Grönwall lemma on the expectation of the average number of particles outside the condensate which controls the distance of the one-body density matrix to the orthogonal projection onto the solution of the Gross-Pitaveskii equation. The next order, i.e. the norm approximation of the $N$ body wave function, requires the study of the fluctuations around the condensate. To do so and following the work of Lewin, Nam, Serfaty and Solovej \cite{LewNamSerSol-13} where the authors analyzed the second order of the ground state energy of a Bose gas, one re-writes the $N$-body Schrödinger evolution in the Fock space of excitations and study their dynamics. In this new setting, one tries to verify Bogoliubov's approximation according to which the evolution of the fluctuation can be obtained by neglecting the terms of order $3$ and $4$ in creation an annihilation operators. This transformation was used in \cite{LewNamSch-15} to prove the norm approximation in the mean-field regime by the solution of the Bogoliubov's equation. In this paper, we follow the same method together with the localization method of \cite{NamNap-17d} where the idea is to use an auxiliary evolution equation defined on the restricted Fock space of at most $M$ excitations. This type of localization in the number of excitations was already present in \cite{LieSol-01,LewNamSerSol-13} and also in \cite{Pickl-11}.

\subsubsection*{\textbf{Acknowledgment}}
This project has received funding from the European Research Council (ERC) under the European Union's Horizon 2020 research and innovation programme (grant agreement MDFT No 725528 of Mathieu Lewin). Part of this work was done when the author benefited from the hospitality  of the Mittag-Leffler Institute, in Stockholm, Sweden.

%%%%%%%%%%%%%%%%%%%%%%%%
%%%%%%%%%%%%%%%%%%%%%%%%
\section{Setting and main result}
%%%%%%%%%%%%%%%%%%%%%%%%
%%%%%%%%%%%%%%%%%%%%%%%%

\subsection{The effective equation}

The purpose of this study is to prove the convergence to the dipolar Gross-Pitaevskii time-dependent equation given by
\begin{equation}\label{eq_gp}
i\partial_t \varphi(t) = \left(-\Delta + a |\varphi(t)|^2 + b K\ast |\varphi(t)|^2 -\mu(t) \right) \varphi(t),
\end{equation}
where $a \in \mathbb{R}$ accounts for the strength of the short range interaction, $b\in \mathbb{R}$ is proportional to the norm of the dipoles and
\begin{equation*}
\mu (t) = \frac{1}{2} \left(a \int_{\mathbb{R}^3} |\varphi(t)|^4 + b \int_{\mathbb{R}^3} K\ast |\varphi(t)|^2 |\varphi(t)|^2 \right).
\end{equation*}
The chemical potential $\mu(t)$ is just a phase factor that we add for convenience but that can be removed by a gauge transformation. The dipolar part is given by
\begin{equation}
	\label{K_general}
	K(x) = \frac{\Omega(x/|x|)}{|x|^3}
\end{equation}
where $\Omega\in L^q(\mathbb{S}^{2})$, for some $q\geq 2$, is a pair function satisfying the following cancellation property on $\mathbb{S}^{2}$, the unit sphere of $\mathbb{R}^{3}$,
\begin{equation}
	\label{cancellation_prop_0}
\int_{\mathbb{S}^2} \Omega(\omega) d\sigma(\omega) = 0,
\end{equation}
with $d\sigma$ denoting the Haar measure on $\mathbb{S}^{2}$. This includes the dipolar potential with $\Omega_{dip}(x) =1-3\cos^2(\theta_x)$ where $\cos(\theta_x) = n\cdot x / |x|$ and where $n$ is a fixed unit vector aligned with all the dipoles. The dipolar interaction is a large distance approximation of a system of Coulomb charges where the size of the dipoles is small compared to the distance between the dipoles. Hence, it is physically relevant to consider interaction looking like $K$ outside of a ball of fixed radius. The convolution with $K$ (in the sense of the principal values) defines a bounded operator in $L^{p}(\mathbb{R}^3)$, $1<p<\infty$ \cite{Duo01} and corresponds to the multiplication in Fourier space by some function $\widehat{K} \in L^\infty(\mathbb{R}^3)$.

\medskip

In order to simplify the computations to come, it is easier to work with the following approximate Gross-Pitaevskii equation
\begin{equation}
	\label{eq:Hartree_bis}
	\left\{
	\begin{array}{l}
	i\partial_t u_{N} = \left(-\Delta  + w_N \ast |u_{N}|^2 - \mu_N(t)\right) u_{N} \\
	u_{N}(0) = u_0,
	\end{array}
	\right.
\end{equation}
where,
\begin{equation*}
\mu_N(t) = \frac{1}{2} \iint_{\mathbb{R}^3\times \mathbb{R}^3} |u_{N}(t,x)|^2 w_N(x-y) |u_{N}(t,y)|^2 \dd x \, \dd y,
\end{equation*}
and $w_N(x) = N^{3\beta} w(N^\beta x)$ for some interaction potential $w : \mathbb{R}^3 \to \mathbb{R}$ and some $\beta > 0$. Choosing
\begin{equation}
	\label{eq:def_w}
w = w_0 + b\mathds{1}_{|x|>R}K
\end{equation}
where $w_0 \in L^1(\mathbb{R}^{3})\cap L^2(\mathbb{R}^{3})$, $b \geq 0$, $R>0$ with $a = \int_{\mathbb{R}^3} w_0$, one can show that the solutions of (\ref{eq_gp}) and (\ref{eq:Hartree_bis}) are close in $L^2$-norm as is stated in \cref{prop:regularity_hartree} below. For the truncated dipolar potential, we also have \cite{Duo01} the existence for all $1< p < \infty$ of some constant $C_p$ independent of $R>0$ such that for all $f\in L^p(\mathbb{R}^3)$, $\|(\mathds{1}_{|x|>R} K)\ast f\|_{L^p(\mathbb{R}^3)} \leq C_p \|f\|_{L^p(\mathbb{R}^3)}$.

 The regularity of the solutions of  (\ref{eq_gp}) and (\ref{eq:Hartree_bis}) has already been well studied. But since (\ref{eq_gp}) depends on $N$, one has to make sure that the Sobolev norms of the solution can be bounded independently of $N$. We do so in the following proposition which is an easy adaptation of \cite[Proposition 3.1]{BenOliSch-15}.
\begin{prop}
	\label{prop:regularity_hartree}
	Let $a,b \in \mathbb{R}$ and let $w$ satisfy (\ref{eq:def_w}), then the Cauchy problems (\ref{eq:Hartree_bis}), respectively (\ref{eq_gp}) (with initial date $u_0$), admit unique maximal solutions respectively $u_{N}, \varphi \in C^1([0,T),L^2(\mathbb{R}^3)) \cap C^0([0,T),H^1(\mathbb{R}^3))$ for some $T>0$.
	If $\widehat{w} \geq 0$ (respectively $a \geq b \inf \widehat{K}$) or if $\|\nabla u_0\|_{L^2(\mathbb{R}^3)}$ is small enough, the solutions $u_{N}$ and $\varphi$ are global in time, $T= +\infty$, and we have the following bounds
\begin{align}
	\label{eq:bound_H1}
\|u_{N}(t)\|_{H^1(\mathbb{R}^3)} + \|\varphi (t)\|_{H^1(\mathbb{R}^3)} &\leq C, \\ 
	\label{eq:bound_Hk}
\|u_{N}(t)\|_{H^k(\mathbb{R}^3)} + \|\varphi(t)\|_{H^k(\mathbb{R}^3)} &\leq C e^{C't}, \textrm{ when moreover } u_0 \in H^{k}(\mathbb{R}^3),
\end{align}
where in the last equation $C$ depends only on $\|u_0\|_{H^k(\mathbb{R}^3)}$ and $C'$ on $\|u_0\|_{H^1(\mathbb{R}^3)}$. \\
	Moreover, we have
	\begin{equation}
	\label{hyp_w_0_td}
\left| \widehat{w_0}(k) - a \right| \leq C |k|,
\end{equation}
for some constant $C>0$, where $a = \int_{\mathbb{R}^3} w_0$, and if $u_0 \in H^{2}(\mathbb{R}^3)$ then
\begin{equation}
\label{eq:CV_GP_GPbis}
\|u_{N}(t) - \varphi(t)\|_{L^2(\mathbb{R}^3)} \leq C \frac{\exp(c_1 \exp(c_2 t))}{N^{\beta}},
\end{equation}
where $C,c_1,c_2 >0$ depend on $\|u_0\|_{H^2(\mathbb{R}^3)}$.
\end{prop}

\begin{rem}
	\label{rem:hypo_CV_potential}
The assumption (\ref{hyp_w_0_td}) is technical and could be reduced with a trade off on the rate of convergence in (\ref{eq:CV_GP_GPbis}). This condition holds for instance as soon as $|x|^2 w_0(x) \in L^1(\mathbb{R}^3)$. Assuming the latter, the parity of $w$ actually implies (\ref{hyp_w_0_td}) with on the right-hand side $|k|^{1+\alpha}$ for all $\alpha<1$. In a similar way, we have automatically
\begin{equation}
	\label{hyp_k_td}
\left|\widehat{\mathds{1}_{|x|\leq R}K} (k) \right|\leq C R^2 k^2
\end{equation}
with a constant $C$ independent of $R$ and $k$. This can be deduced from the following formula \cite[Lemma 9]{Tri-18}
\begin{equation*}
\widehat{\mathds{1}_{|x|\leq R}K} (k) = \int_{\mathbb{S}^2} \int_0^{R} \frac{\cos(r k \cdot \omega) - 1}{r} \, \Omega(\omega) \, \mathrm{d}r \, \mathrm{d}\sigma(\omega).
\end{equation*}
\end{rem}

\begin{proof}[Proof of \cref{prop:regularity_hartree}]
The existence and uniqueness are standard, see \cite{Cazenave}, and comes from the regularity properties of the convolution with $K$ \cite{stein1970singular}. We also have from usual techniques the blow-up alternative, that is if $T<\infty$ then $\|u(t)\|_{H^1(\mathbb{R}^3)} \to \infty$ as $t\to T$. Hence if $\widehat{w} \geq 0$ then we have
\begin{equation*}
C = \int_{\mathbb{R}^3} |\nabla u_{N}(t)|^2 + \iint_{\mathbb{R}^3\times \mathbb{R}^3} |u_{N}(t)|^2 w_N\ast |u_{N}(t)|^2 \geq \|\nabla u_{N}(t)\|_{L^2(\mathbb{R}^3)}^2,
\end{equation*}
and similarly for $\varphi(t)$. From this and the blow-up alternative we deduce global existence in this case. If $\|\nabla u_0\|_{L^2(\mathbb{R}^3)}$ is small enough, it is also standard that $\|u_{N}(t)\|_{H^1(\mathbb{R}^3)}$ and $\|\varphi(t)\|_{H^1(\mathbb{R}^3)}$ have to remain bounded \cite{CarMarkSpa-08,BahGer-99}. Hence the global existence yields in this case too.

 The bounds on the growth of $\|u_{N}(t)\|_{H^k(\mathbb{R}^3)}$ and  $ \|\varphi(t)\|_{H^k(\mathbb{R}^3)} $ are obtained via the same proof of \cite[Proposition 3.1]{BenOliSch-15} where the authors only used that $w\in L^{p}(\mathbb{R}^3)$, for $p>1$ and $\widehat{w} \in L^{\infty}(\mathbb{R}^3)$. Finally, the bound on $\|u_{N}(t) - \varphi(t)\|_{L^2(\mathbb{R}^3)}$ is also obtained the same way as in \cite[Proposition 3.1]{BenOliSch-15} using (\ref{hyp_w_0_td}) and \cref{rem:hypo_CV_potential}.
\end{proof}

\subsection{Main result: derivation of the Gross-Pitaevskii equation}

We consider $N$ bosons in $\mathbb{R}^{3}$ interacting via a pair potential $w_N(x-y) := N^{3 \beta}w(N^{\beta}(x-y))$, where $w \in L^{6/5}(\mathbb{R}^3) \cap L^2(\mathbb{R}^3)$ is such that $\widehat{w}\in L^\infty(\mathbb{R}^3)$. Note that the interaction potential is possibly long-range and is allowed to be (partly) attractive. The system is entirely described at any time $t$ by its wave function $\Psi_N(t)$ evolving in $L^2(\mathbb{R}^{3})^{\otimes_s N}$, the symmetric tensor product of $N$ copies of $L^2(\mathbb{R}^{3})$, whose dynamics is given by the Schrödinger equation
\begin{equation}\label{eq:schro}
i \partial_t \Psi_N = H_N \Psi_N.
\end{equation}
Here $H_N$ is the Hamiltonian of the system given by
\begin{equation}\label{H_N}
	H_N = \sum_{j=1}^N -\Delta_{x_j} + \frac{1}{N-1} \sum_{1\leq i < j \leq N} w_N(x_i-x_j).
\end{equation}

Even though we are interested in the dynamics, the behavior of the ground state energy plays an important role. In such a system where the interaction has a negative part, proving stability of the second kind (that is, $H_N \geq -C N$ for some constant $C>0$ independent of $N$) is a difficult problem. In \cite{Tri-18}, it was proven that in the presence of an external confining potential $V(x) \geq C |x|^s$ for some $s>0$, if $w$ satisfies (\ref{eq:def_w}) and $\beta < 1/3 + s/(45+42s)$ the Hamiltonian $H_N$ is stable of the second kind. Note that with the only assumption that $\widehat{w}\geq 0$, $\widehat{w} \in L^1(\mathbb{R}^3)$ we automatically have $H_N \geq -C N$ for all $\beta\leq 1/3$.

The goal of this paper, loosely speaking, is to show that if the initial wave function $\Psi_N(0)$ is close to a product state $\varphi(0)^{\otimes N}$, then the propagated wave function $\Psi_N(t)$ remains well approximated by the product state $\varphi(t)^{\otimes N}$, where $\varphi(t)$ solves the non-linear Gross-Pitaevskii equation (\ref{eq_gp}). This \emph{approximation} is true at first order, that is we can prove that the one-body reduced density matrix of $\Psi_{N}(t)$ converges towards the orthogonal projector onto $\varphi(t)$. But this fails as a norm approximation in $L^2(\mathbb{R}^{3N})$ since $\Psi_N$ is never a pure condensate and contains fluctuations around it. The dynamics of these fluctuations is encoded in the time-dependent Bogoliubov equation that we define later on.

Before stating our main result, we recall the definition of the $k$-particle reduced density matrix of a pure state $\Psi \in L^2(\mathbb{R}^{3})^{\otimes_s N}$, 
\begin{equation}
\label{eq:def:reduced_density_matrix}
\Gamma^{(k)}_{\Psi} := \tr_{k+1 \to N} \ket{\Psi}\bra{\Psi}
\end{equation}
where $\ket{\Psi}\bra{\Psi}$ is the orthogonal projection onto $\Psi$, or in terms of kernel,
\begin{equation*}
\Gamma^{(k)}_{\Psi}(x_1,...,x_k,y_1,...,y_k) = \int \Psi(x_1,...,x_k,z_{k+1},...,z_N) \overline{\Psi(y_1,...,y_k,z_{k+1},...,z_N)} \, \mathrm{d} z_{k+1}... \, \mathrm{d}z_N.
\end{equation*}
%For shortness, when we consider a sequence $(\Psi_{N,t})_N$, for some $t\in\mathbb{R}$, we will denote by $\Gamma_{N,t}^{(k)}$ its $k$-particle reduced density matrix. 

In the following we denote by $\Psi_N(t,x)$ the solution to (\ref{eq:schro}) and $u_{N}(t,x)$ the solution to (\ref{eq:Hartree_bis}).
%For shortness we will write from now on $\Psi_N \equiv \Psi_N(t,\cdot)$ and $\varphi \equiv u(t,\cdot)$ and we omit the $t$ dependence. 
We will denote by $\mathcal{H} := L^2(\mathbb{R}^3)$ the one particle space and by $\mathcal{H}^N := L^2(\mathbb{R}^3)^{\otimes_s N}$ the $N$ particle bosonic space. We define $P(t) = \ket{u_{N}(t)}\bra{u_{N}(t)}$ the orthogonal projector onto $u_{N}(t) \in L^2(\mathbb{R}^3)$ and $Q (t)= 1 - P(t)$.

To describe the excitations orthogonal to the condensate, we follow the technique of \cite{LewNamSerSol-13,LewNamSch-15}. Let us denote by $\mathcal{H}_+ = \{u_N\}^\perp$ the orthogonal space of $u_N$ in $\mathcal{H}$, then note that the $N$-body wave function $\Psi_N$ admits the unique decomposition
\begin{equation*}
\Psi_N = u_N^{\otimes N} \varphi_0 + u_N^{\otimes N-1} \otimes_s \varphi_1 + u_N^{\otimes N-2} \otimes_s \varphi_2 + ... + \varphi_N
\end{equation*}
with $\varphi_k \in \mathcal{H}_+^{\otimes_s k}$ for all $k\geq 0$ with the convention that $\varphi_0 \in \mathbb{C}$. The above decomposition allows to define the unitary map 
\begin{equation}
	\label{def:U_N}
	\begin{array}{ccc}
		U_N : &\mathcal{H}^N \longrightarrow &\mathcal{F}(\mathcal{H}_+) \\
		&\Psi_N \longmapsto & \Phi_N := \bigoplus_{k=0}^N \varphi_k.
	\end{array}
\end{equation}
The unitary transformation $U_N$ is a one-to-one correspondance between a $N$ particle state and its excitations orthogonal to the condensate.

In the sequel we will denote by $\Phi_N(t) = U_N(t) \Psi_N(t)$ the corresponding state describing the excitations in the truncated Fock space $\mathcal{F}^{\leq N}(\mathcal{H}_+)$. The Bogoliubov approximation consists in approximating $\Phi_N(t)$ by the solution $\Phi(t)$ of the time dependent Bogoliubov equation
\begin{equation}
	\label{eq:Bog_equation}
\left\{
\begin{array}{l}
i \partial_t \Phi(t) = \mathbb{H}(t) \Phi(t) \\
\Phi(0) = U_N(0) \Psi_N(0),
\end{array}
\right.
\end{equation}
where the Bogoliubov Hamiltonian $\mathbb{H}$ is the operator acting on the entire Fock space $\mathcal{F}(\mathcal{H})$ (not only $\mathcal{F}(\mathcal{H}_+)$) given by
\begin{equation*}
\mathbb{H}(t) = \dd \Gamma(h(t)) + \frac{1}{2} \iint_{\mathbb{R}^3 \times \mathbb{R}^3} \left(K_2(t,x,y) a^*_x a^*_y + \overline{K_2(t,x,y)}a_x a_y \right) \dd x\, \dd y\, .
\end{equation*}
Where
\begin{equation*}
h(t) = -\Delta +  w_N \ast |u_N(t,x)|^2 + Q(t) \widetilde{K}_1(t) Q(t)  - \mu_N(t) , \qquad K_2(t) = Q(t) \otimes Q(t) \widetilde{K}_2(t) \in \mathcal{H}^2,
\end{equation*}
with $\widetilde{K}_1(t) \in \mathcal{B}(\mathcal{H})$ is the operator of kernel $$\widetilde{K}_1(t)(x,y) = u_N(t,x)w_N(x-y) \overline{u_N(t,y)},$$ and $\widetilde{K}_2(t) \in \mathcal{H}^2$ is the function given by $$\widetilde{K}_2(t,x,y) =  u_N(t,x)w_N(x-y) u_N(t,y).$$ It was proven in \cite{LewNamSch-15} that the Cauchy problem (\ref{eq:Bog_equation}) is well posed when $\langle \Phi(0), \dd \Gamma(1-\Delta) \Phi(0) \rangle < \infty$ and when $u_N \in C^0([0,T),H^1(\mathbb{R}^3)) \cap C^1([0,T),H^{-1}(\mathbb{R}^3))$. Then there is unique corresponding solution of (\ref{eq:Bog_equation}), $\Phi \in C^0([0,T],\mathcal{F}(\mathcal{H})) \cap L^\infty_{\loc}([0,T],\mathcal{Q}(\dd \Gamma(1-\Delta)))$. We emphasize that even if the Cauchy problem (\ref{eq:Bog_equation}) is posed in $\mathcal{F}(\mathcal{H})$, the solution satisfies $\Phi(t) \in \mathcal{F}(\mathcal{H}_+)$ as one can verify by computing the time-derivative of the quantity $\|a(u_N(t)) \Phi(t)\|_{L^2(\mathbb{R}^{3N})}$ and observe that it vanishes.

We recall that the norm approximation of $\Psi_N(t)$ by $U_N^* \Phi(t)$ is stronger than the convergence of the one body reduced density matrix $\Gamma^{(1)}$ towards $\ket{u_N(t)}\bra{u_N(t)}$. This is why in our result below the range of validity for the parameter $\beta$ is wider when we look at the convergence of the one-body reduced density matrix.

We can now state our main result.

\begin{theo}[Main Theorem]\label{theo_1}
Let $\beta >0$ and let $w = w_0 + b\mathds{1}_{|x|>R}K$ where $w_0 \in L^1(\mathbb{R}^{3})\cap L^2(\mathbb{R}^{3})$, $b \geq 0$, $R>0$ and where $K$ is given by (\ref{K_general}). Let $u_{N}$ be a solution of the Gross-Pitaevskii equation (\ref{eq:Hartree_bis}) on some interval $[0,T)$ with $T\in \mathbb{R}_+\cup \{\infty\}$ such that (\ref{eq:bound_H1}) and (\ref{eq:bound_Hk}) hold. Let $(\Psi_N(0))_N$ be such that 
\begin{equation}
	\label{hypo_kinetic_excitations}
\tr ((-\Delta)^{1/2} Q(0) \Gamma^{(1)}_{\Psi_N(0)} Q(0)(-\Delta)^{1/2} ) \leq C_0 N^{-1}
\end{equation}
for some constant $C_0>0$ and $\Psi_N(t)$ be the solution to the Schrödinger equation (\ref{eq:schro}) with initial condition $\Psi_N(0)$. Let $\Phi(t) = (\varphi_k(t))_{k\geq 0}$ be the solution of the Bogoliubov equation (\ref{eq:Bog_equation}).
\begin{enumerate}
\item If $0 < \beta < 1/6$ then for all $0<\alpha < \min((1-6\beta)/4,(2-7\beta)/4)$ we have 
\begin{equation}
	\label{eq:cv_norm}
\left\|\Psi_N(t) - \sum_{k=0}^N u_N(t)^{\otimes k} \otimes_s \varphi_k (t) \right\|_{L^2(\mathbb{R}^{3N})}^2 \leq C_{\alpha}e^{C't} N^{-\alpha},
\end{equation}
where $C_\alpha$ depends on $\alpha$, $C_0$ and $\|u(0)\|_{H^4(\mathbb{R}^3)}$ and where $C'$ depends on $\|u(0)\|_{H^1(\mathbb{R}^3)}$.
\item If $0< \beta < 1/4$ then for all $0 < \alpha < \min( (3 - 10\beta)/4, (1-4\beta)/4)$ we have
\begin{equation}
	\label{eq:cv_trace_class}
\| \Gamma^{(1)}_{\Psi_N(t)} - \ket{u_N(t)}\bra{u_N(t)} \|_{\mathfrak{S}_1} \leq C_{\alpha}e^{C' t} N^{-\alpha},
\end{equation}
where $C_\alpha$ depends on $\alpha$, $C_0$, and $\|u(0)\|_{H^4(\mathbb{R}^3)}$ and $C'$ depends on $\|u(0)\|_{H^1(\mathbb{R}^3)}$.
\item
Let moreover assume that $\widehat{w} \geq 0$ and that $\widehat{w} \in L^1(\mathbb{R}^3)$.
\begin{enumerate}
\item
If $0<\beta < 1/5$ then we have (\ref{eq:cv_norm}) for $0< \alpha < \min((3-10\beta)/4,(1- 5\beta)/4) $.
\item
If $0<\beta < 3/8$ then we have (\ref{eq:cv_trace_class}) for $0< \alpha < \min((1-\beta)/2, (2-5\beta)/2, (3-8\beta)/4)$.
\end{enumerate}
\end{enumerate}
\end{theo}

\begin{rem}
As said earlier, \emph{a priori} estimates on the kinetic energy are crucial in our proof. This is why the assumption $\widehat{w} \geq 0$ allows us to extend the range of $\beta$. Assuming stability of the second kind, one could improve it again.
\end{rem}

\begin{rem}
Note that the range of $\beta$ includes regimes where the stability of the second kind is not established. In particular, when $\widehat{w}\geq 0$ we allow $0 < \beta < 3/8$ which is above the threshold $1/3$. In such a regime $\beta >1/3$, the system is very dilute since the range of the interaction is much smaller than the mean distance between particles.
\end{rem}

The rest of the paper is devoted to the proof of \cref{theo_1}.

\section{The localization method}

\subsection{Presentation of the method}
To simplify notations, we will in the sequel denote by $u_N(t)$ the solution of modified Gross-Pitaevskii equation (\ref{eq:Hartree_bis}). The Schrödinger evolution (\ref{eq:schro}) is unitarily equivalent to the following dynamics for the excitations outside the condensate. Recalling that $\Phi_N(t) = U_N(t) \Psi_N(t)$, we have
\begin{align*}
\left\{
\begin{array}{l}
i\partial_t \Phi_N(t) = \mathcal{G}_N(t) \Phi_N(t) \\
\quad \,\Phi_N(0) = U_N(0) \Psi_N(0),
\end{array}
\right.
\end{align*}
with 
\begin{align*}
\mathcal{G}_N(t) = \mathds{1}^{\leq N}\left(\mathbb{H}(t) + \mathcal{E}_N(t)\right)\mathds{1}^{\leq N} = U_N H_N U_N^*
\end{align*}
and $\mathcal{E}_N(t)$ is an error term which is given by
\begin{equation*}
\mathcal{E}_N(t) = \frac{1}{2} \sum_{j=0}^4 (R_j + R_j^*),
\end{equation*}
\begin{align*}
R_0 &= R_0^* = d\Gamma (Q(t) [w_N \ast |u_N(t)|^2 + \widetilde{K}_1(t) - \mu_N(t)]Q(t)) \frac{1-\mathcal{N}}{N-1}, \\
R_1 &= -2 \frac{\mathcal{N}\sqrt{N-\mathcal{N}}}{N-1} a(Q(t) [w_N\ast |u_N(t)|^2] u_N(t)), \\
R_2 &= \iint K_2(t,x,y) a^*_x a^*_y \dd x \, \dd y \, \left( \frac{\sqrt{(N-\mathcal{N})(N-\mathcal{N}-1)}}{N-1} - 1\right) \\
R_3 &= \frac{\sqrt{N-\mathcal{N}}}{N-1} \iiiint (1\otimes Q(t) w_N Q(t) \otimes Q(t))(x,y,x',y') \overline{u(t,x)} a^*_y a_{x'} a_{y'}  \dd x \, \dd y \,\dd x' \,\dd y' \, , \\
R_4 &= R_4^* = \frac{1}{2(N-1)} \iiiint (Q(t) \otimes Q(t) w_N Q(t) \otimes Q(t) ) (x,y,x',y') a^*_x a^*_y a_{x'} a_{y'} \dd x \, \dd y \,\dd x' \,\dd y' \,.
\end{align*}
This computation can be found in \cite{LewNamSch-15}. Besides this reformulation of the Schrödinger equation, we will use the localization method which consists in using an auxiliary dynamics localized in the truncated Fock space $\mathcal{F}^{\leq M}(\mathcal{H}_+)$ for $M = N^{1-\delta}$, for some $\delta >0$. Having an \emph{a priori} bound on the number of excitations allows to control accurately the error terms above and hence also the expectation of the kinetic energy $\dd \Gamma(1-\Delta)$, which itself controls the expectation of the  number of excitations $\mathcal{N}_+$. More precisely, the localized dynamics is given, for $1 \leq M \leq N$, by
\begin{equation}
	\label{eq:truncated_dynamics}
\left\{
\begin{array}{l}
i\partial_t \Phi_{N,M}(t) = \mathds{1}^{\leq M} \mathcal{G}_N(t) \mathds{1}^{\leq M} \Phi_{N,M}(t) \\
\quad \,\Phi_{N,M}(0) = U_N(0) \Psi_N(0).
\end{array}
\right.
\end{equation}
We have denoted by $\mathds{1}^{\leq M} := \mathds{1}\left(\mathcal{N} \leq M\right)$ the spectral projection associated to the number operator $\mathcal{N}$. The existence and uniqueness of the solution of (\ref{eq:truncated_dynamics}) follows from \cite[Theorem 7]{LewNamSch-15}. Here as well, a direct computation shows that the time derivative of $\|a(u_N(t)) \Phi_{N,M}(t)\|_{L^2(\mathbb{R}^3N)}^2$, $\|a(u_N(t)) \Phi_{N}(t)\|_{L^2(\mathbb{R}^3N)}^2$ and $\|\mathds{1}^{\leq M}\Phi_{N,M}(t)\|_{L^2(\mathbb{R}^3N)}^2$ vanish, see \cite{NamNap-17d,LewNamSch-15}, implying that $\Phi_N(t) \in \mathcal{F}(\mathcal{H}_+)$ and $\Phi_{N,M}(t) \in \mathcal{F}^{\leq M}(\mathcal{H}_+)$ for all $t\geq 0$.

\subsection{Estimate on the kinetic energy}
We start by proving that the assumption (\ref{hypo_kinetic_excitations}) is enough to bound the whole energy of $\Psi_N(t)$.
\begin{prop}
Assume (\ref{hypo_kinetic_excitations}) then for all $t \in [0,T)$
\begin{equation*}
\langle \Psi_N(t), H_N \Psi_N(t) \rangle = \langle \Psi_N(0), H_N \Psi_N(0) \rangle \leq C N.
\end{equation*}
\end{prop}

\begin{proof}
The equality follows from differentiating $\langle \Psi_N(t), H_N \Psi_N(t) \rangle$ and the use of (\ref{eq:schro}). We focus on proving the inequality.
 By the Cauchy-Schwartz inequality we have that
\begin{align}
\tr_{\mathcal{H}} (-\Delta \Gamma^{(1)}_{\Psi_N}) \nn
	&\leq 2 \tr_{\mathcal{H}} (P(0) (-\Delta) P(0) \Gamma^{(1)}_{\Psi_N}) + 2 \tr_{\mathcal{H}} (Q(0) (-\Delta) Q(0) \Gamma^{(1)}_{\Psi_N}) \nn \\
	&\leq 2 \|\nabla u(0)\|_{L^2(\mathbb{R}^3)}^2 + 2 \frac{C_0}{N} \leq C,
		\label{ineq:borne_kinetic_N_body}
\end{align}
where we have used assumption (\ref{hypo_kinetic_excitations}).
Similarly
\begin{align*}
w_N(x-y) &= P(0)\otimes1 w_N(x-y) P(0)\otimes1 + P(0) \otimes1 w_N(x-y) Q(0) \otimes1 \\
&\qquad \qquad \qquad \qquad \qquad  
 + Q(0) \otimes1 w_N(x-y) P (0)\otimes1 + Q(0) \otimes1 w_N(x-y) Q(0) \otimes1 \\
	&\leq w_N \ast |u_N(t)|^2 (y) + \eta |w_N| \ast |u_N(t)|^2 (y) + (1+\eta^{-1}) Q(0) \otimes1 |w_N(x-y)| Q(0) \otimes1 \\
	&\leq 2 \Big( ( \| w_N \ast |u_N(t)|^2\|_{L^{3/2}(\mathbb{R}^3)} + \eta \| w_N \ast |u_N(t)|^2\|_{L^{3/2}(\mathbb{R}^3)}) (-\Delta_y) \\
	&\qquad \qquad \qquad \qquad \qquad 
	 + (1+\eta^{-1})\|w_N\|_{L^{3/2}(\mathbb{R}^3)} Q(0)\otimes1 (-\Delta_{x}) Q(0)\otimes1\Big)
\end{align*}
for all $\eta >0$, where we used Sobolev's inequality in the last inequality. Since $w$ satisfies (\ref{eq:def_w}), we have 
$$\| w_N \ast |u_N(t)|^2\|_{L^{3/2}} \leq (\|w_0\|_{L^1(\mathbb{R}^3)} + C) \|u_N(t)\|_{L^{3}(\mathbb{R}^3)}^2 \leq C$$ 
for some constant $C>0$ depending only on $\|u(0)\|_{H^1(\mathbb{R}^3)}$ and 
$$\| |w_N| \ast |u_N(t)|^2\|_{L^{3/2}} \leq \|w_N\|_{L^{1+\varepsilon}(\mathbb{R}^3)} \|u_N\|_{L^{3-2\varepsilon'}(\mathbb{R}^3)}^2 \leq C N^{3\beta \varepsilon/(1+\varepsilon)} \|w\|_{L^{1+\varepsilon}(\mathbb{R}^3)}$$
where $\varepsilon, \varepsilon' >0$, $\varepsilon' \leq 1/2 $ are such that
\begin{equation*}
\frac{1}{1+\varepsilon} + \frac{1}{3/2-\varepsilon'} = 1 + \frac{2}{3},
\end{equation*}
and where $C>0$ is some other constant depending only  $\|u(0)\|_{H^1(\mathbb{R}^3)}$. Hence we obtain
\begin{equation*}
w_N(x-y) \leq C (1+ \eta N^{3\beta \varepsilon/(1+\varepsilon)} ) (-\Delta_y) + (1+ \eta^{-1}) Q(0)\otimes 1 (-\Delta_x) Q(0)\otimes 1,
\end{equation*}
for all $\varepsilon, \eta >0$. From this, we deduce
\begin{align*}
\langle \Psi_N H_N \Psi_N \rangle 
	&= N \left( \tr_{\mathcal{H}} (-\Delta \Gamma^{(1)}_{\Psi_N}) + \tr_{\mathcal{H}^{2}} (w_N(x-y) \Gamma^{(2)}_{\Psi_N}\right) \\
	&\leq 2 N  \tr_{\mathcal{H}} (-\Delta \Gamma^{(1)}_{\Psi_N})  + C N\tr_{\mathcal{H}^{2}} ((1+ \eta N^{3\beta \varepsilon/(1+\varepsilon)} ) (-\Delta_y)  \Gamma^{(2)}_{\Psi_N}) \\
	&\quad  + C \tr_{\mathcal{H}^{2}} ((1+ \eta^{-1} Q(0)\otimes 1 (-\Delta_x) Q(0)\otimes 1)\Gamma^{(2)}_{\Psi_N}) \\
	&\leq CN(1 + \eta N^{3\beta \varepsilon/(1+\varepsilon)})  \tr_{\mathcal{H}} (-\Delta \Gamma^{(1)}_{\Psi_N}) + (1+\eta^{-1})\tr_{\mathcal{H}} (Q(0) (-\Delta) Q(0) \Gamma^{(1)}_{\Psi_N})
\end{align*}
for all $\varepsilon, \eta >0$. Now using the assumption (\ref{ineq:borne_kinetic_N_body}) and (\ref{hypo_kinetic_excitations}) we obain
\begin{align*}
\langle \Psi_N H_N \Psi_N \rangle \leq  CN (1 + \eta N^{3\beta \varepsilon/(1+\varepsilon)}) + (1+\eta^{-1}) C_0 N^{-1}
\end{align*}
and taking $\eta = N^{-3\beta \varepsilon/(1+\varepsilon)}$ for $\varepsilon \leq 3\beta$ proves the result.
\end{proof}

Denote by
\begin{equation*}
e_N = \inf \sigma\left(\frac{1}{2}\sum_{k=1}^N -\Delta_{x_k} + \frac{1}{N-1} \sum_{1\leq j<k \leq N} w_N(x_j - x_k) \right).
\end{equation*}
A computation shows the following.
\begin{cor} For all $t \in [0,T)$,
\begin{align}
	\label{eq:bound_kinetic_true_state}
\langle \Phi_N(t), \dd \Gamma(1-\Delta) \Phi_N(t) \rangle 
	\leq 2 \langle \Psi_N(t), H_N \Psi_N(t) \rangle - 2 e_N \leq C (|e_N| + N),
\end{align}
\end{cor}
 We will need the following two intermediate results. But first, we establish the convention that throughout the rest of the paper $C_{t,\varepsilon}$ denotes a constant that can always be bounded by $C_{\varepsilon} e^{C't}$ where $C_{\varepsilon}$ depends on $\varepsilon >0$ and $\|u(0)\|_{H^4(\mathbb{R}^3)}$ and $C'>0$ depends on $\|u(0)\|_{H^1(\mathbb{R}^3)}$. It is easily verified in the proofs using \cref{prop:regularity_hartree}. In this way we can retrieve the right hand-side of the estimates of \cref{theo_1}.

\begin{prop}[Kinetic estimate for the truncated dynamics]
	\label{prop:truncated_stability}
Let $0< \beta < 1$ and $M= N^\alpha$ with $0 < \alpha < 1-\beta$, let $\varepsilon>0$, Then, for $N$ large enough, we have
\begin{equation*}
\langle \Phi_{N,M}(t) , \dd \Gamma(1-\Delta) \Phi_{N,M}(t)\rangle \leq C_{t,\varepsilon} N^{\beta + \varepsilon},
\end{equation*}
for some constant $C_{t,\varepsilon}>0$.
\end{prop}

\begin{prop}[Kinetic estimate for the Bogoliubov dynamics]
	\label{prop:Bog_stability}
Let $0< \beta < 1$ and let $\varepsilon>0$, we have
\begin{equation*}
\langle \Phi (t) , \dd \Gamma(1-\Delta) \Phi (t)\rangle \leq C_{t,\varepsilon} N^{\beta + \varepsilon},
\end{equation*}
for some constant $C_{t,\varepsilon}>0$.
\end{prop}

The proofs of \cref{prop:truncated_stability} and \cref{prop:Bog_stability} are similar and we will only give the one of  \cref{prop:truncated_stability}. It is a consequence of the following lemmas whose proofs are postponed until the end of the proof of Proposition \ref{prop:Bog_stability}.

\begin{lem}
	\label{lem:stability_H}
Let $\beta,\varepsilon > 0$. There exists some constant $C_{t,\varepsilon} >0$ such that for $A\in \{ \mathbb{H}(t) + \dd \Gamma(\Delta), \partial_t \mathbb{H}(t), i [\mathbb{H},\mathcal{N}]\}$, we have
\begin{equation*}
\pm A\leq C_{t,\varepsilon}\left( \mathcal{N} + \eta N^{\beta+ \varepsilon} + \eta^{-1} \dd \Gamma(1-\Delta)\right),
\end{equation*}
for any $\eta >0$.
\end{lem}

\begin{lem}
	\label{lem:Bog_approx}
Let $0 < \beta < 1$. For every $1 \leq m \leq N$, $\varepsilon >0$ and for $A \in \{ \mathcal{E}_N(t), \partial_t \mathcal{E}_N(t), i[\mathcal{E}_N(t),\mathcal{N}]\}$ we have
\begin{equation}
	\label{eq:lem_Bog_approx}
\pm \mathds{1}^{\leq m} A \, \mathds{1}^{\leq m} \leq C_{t,\varepsilon} \sqrt{\frac{m}{N^{1-\beta-\varepsilon}}} \dd \Gamma(1-\Delta) \, \textrm{ in } \mathcal{F}(\mathcal{H}_+).
\end{equation}
%and
%\begin{equation}
%	\label{eq:lem_Bog_approx_2}
%\pm \mathds{1}^{\leq m} A \, \mathds{1}^{\leq m} \leq C_{t,\varepsilon} \frac{m}{N^{1-\beta}} \dd \Gamma(1-\Delta) + \frac{m}{N^{(1-\beta - \varepsilon)/2}}\left(\eta + \eta^{-1} \dd\Gamma(1-\Delta) \right) \, \textrm{ in } \mathcal{F}(\mathcal{H}_+),
%\end{equation}
%for all $\eta >0$ and 
%where $C_{t,\varepsilon}$ is a constant that depends polynomially on $\|u_N(t)\|_{H^4(\mathbb{R}^3)}$, $\|\widehat{w}\|_{L^\infty(\mathbb{R}^3)}$ and $\|w\|_{L^{2}(\mathbb{R}^3)}$.
\end{lem}

\begin{proof}[Proof of \cref{prop:truncated_stability}]
Let $\varepsilon >0$ and define $\varepsilon_\alpha = (1-\beta - \alpha)/2$. We use \cref{lem:Bog_approx} with $\varepsilon_\alpha$ and \cref{lem:stability_H} with $\varepsilon$. For $N$ large enough we have 
\begin{equation*}
A(t) := C_{t,\varepsilon} \mathds{1}^{\leq M}\left(\mathcal{N} + N^{\beta + \varepsilon}\right) + \mathds{1}^{\leq M} \mathcal{G}_N(t) \mathds{1}^{\leq M} \geq \frac{1}{2} \mathds{1}^{\leq M} \dd \Gamma(1-\Delta),
\end{equation*}
hence we only need to control the left-hand side to show the proposition. For this we use Grönwall's lemma
\begin{align*}
\frac{\dd}{\dd t} \langle \Phi_{N,M}(t), A(t) &\Phi_{N,M}(t) \rangle \\
&= \langle \Phi_{N,M}(t), \partial_t A(t) \Phi_{N,M}(t) \rangle + i \langle \Phi_{N,M}(t), [\mathds{1}^{\leq M} \mathcal{G}_N(t) \mathds{1}^{\leq M}, A(t)] \Phi_{N,M}(t) \rangle \\
&= \langle \Phi_{N,M}(t),  (\partial_t \mathbb{H}(t) + \partial_t \mathcal{E}_N(t) + \partial_t C_{t,\varepsilon} (N^{\beta+ \varepsilon } + \mathcal{N})) \Phi_{N,M}(t) \rangle  \\
&\quad + i \langle \Phi_{N,M}(t), [\mathds{1}^{\leq M} \mathcal{G}_N(t) \mathds{1}^{\leq M}, C_{t,\varepsilon} \mathcal{N}] \Phi_{N,M}(t) \rangle \\
&\leq C_{t,\varepsilon} \langle \Phi_{N,M}(t), A(t) \Phi_{N,M}(t) \rangle,
\end{align*}
where we used \cref{lem:stability_H} and \cref{lem:Bog_approx} in the last inequality.
\end{proof}

The proofs of \cref{lem:stability_H} and \cref{lem:Bog_approx} require the following intermediate results. The first one, \cref{lem:Bog_lower_bound} is a slightly adapted version of \cite[Lemma 9]{NamNapSol-16}. And the second one, \cref{lem:three_body} is a estimate on the three body term of the error in the Bogoliubov approximation.

\begin{lem}
	\label{lem:Bog_lower_bound}
Let $H>0$ be a self-adjoint operator on $\mathcal{H}$. Let $K:\overline{\mathfrak{H}} \equiv\mathfrak{H}^* \to \mathfrak{H}$ be an operator with kernel $K(x,y) \in \mathfrak{H}^2$. Assume that $KH^{-1}K^* \leq H$ and that $H^{-1/2} K$ is Hilbert-Schmidt. Let $\chi_1, \chi_2 :\mathbb{R} \to [0,1] $ then
\begin{align}
	\label{eq:lem_def_Ham}
\widetilde{\mathbb{H}} := \dd \Gamma(H) + \frac{1}{2} \iint \Big(K(x,y) \chi_1(\mathcal{N})  a^*_x a^*_y \chi_2(\mathcal{N}) + \overline{K(x,y)}  &\chi_1(\mathcal{N})a_x a_y \chi_2(\mathcal{N}) \Big) \, \dd x \, \dd  y \nn \\ 
& \geq - \frac{1}{2}\|H^{-1/2} K\|^2_{\textrm{HS}}.
\end{align}
\end{lem}
\begin{proof}
The proof is a simple adaptation of the one of \cite{NamNapSol-16}. For $\Psi \in D(\widetilde{\mathbb{H}})$ we have
\begin{equation*}
\Big\langle \Psi \widetilde{\mathbb{H}} \Psi \Big\rangle = \tr( H^{1/2} \gamma_\Psi H^{1/2} ) + \Re \tr (K^*\widetilde{\alpha}_\Psi),
\end{equation*}
where the operators $\gamma_\Psi : \mathfrak{H} \to \mathfrak{H}$ and $\widetilde{\alpha}_\Psi : \mathfrak{H} \to \mathfrak{H}^*$ are defined in the following way. For all $f,g \in \mathfrak{H}$,
\begin{equation*}
\langle f, \gamma_\Psi g \rangle = \langle \Psi , a^*(g)a(f) \Psi \rangle, \quad \langle f, \widetilde{\alpha}_\Psi g \rangle = \langle \Psi , \chi_1(\mathcal{N}) a^*(g)a^*(f)\chi_2(\mathcal{N}) \Psi \rangle.
\end{equation*}
Here $a^*(f)$ and $a(f)$ for $f\in \mathfrak{H}$ are the creation and annihilation operators in Fock space. Let $J : \mathfrak{H} \to \mathfrak{H}^*$ defined by $J(f)(g) = \braket{f,g}$ for all $f,g \in \mathfrak{H}$, then we have
\begin{align*}
\Bigg \langle
\begin{pmatrix}
f \\
Jg
\end{pmatrix},&
\begin{pmatrix}
\gamma_\Psi &  \widetilde{\alpha}_\Psi^* \\
 \widetilde{\alpha}_\Psi & 1 + J\gamma_\Psi J^*
\end{pmatrix}
\begin{pmatrix}
f \\
Jg
\end{pmatrix}
\Bigg\rangle \\
&=
\Bigg\langle \Psi ,\Bigg\{ \bigg(\chi_1(\mathcal{N}) a^{*}(g) + \chi_2(\mathcal{N}) a(f) \bigg) \bigg((a(g)\chi_1(\mathcal{N})+ a^*(f)\chi_2(\mathcal{N})\bigg) \\
& \qquad \qquad \qquad 
+ (1-\chi_1(\mathcal{N})^2) a^*(g)a(g) + (1- \chi_2(\mathcal{N})^2)  a^*(f)a(f) \Bigg\} \Psi \Bigg\rangle \geq 0.
\end{align*}
Hence we have by \cite[Lemma 3]{NamNapSol-16} that
\begin{equation}
	\label{eq:lem_Bog_ham}
\gamma_\Psi \geq 0 \quad \textrm{ and } \quad  \gamma_\Psi \geq \widetilde{\alpha}_\Psi^* (1 + J\gamma_\Psi J^* )^{-1} \widetilde{\alpha}_\Psi.
\end{equation}
The rest of the proof proceeds as in \cite{NamNapSol-16}. We only need to prove it for $\Psi$ such that $\gamma_\Phi$ and $ \widetilde{\alpha}_\Psi$ are finite-rank operator, in which case we have
\begin{align*}
\left| \tr(K^* \widetilde{\alpha}_\Psi) \right| 
&= \left| \tr( \widetilde{\alpha}_\Psi K^*) \right| 
= \left| \tr((1 + J\gamma_\Psi J^* )^{-1/2} \widetilde{\alpha}_\Psi H^{1/2} H^{-1/2} K^* (1 + J\gamma_\Psi J^* )^{1/2}) \right| \\
&\leq \| (1 + J\gamma_\Psi J^* )^{-1/2} \widetilde{\alpha}_\Psi H^{1/2} \|_{\textrm{HS}} \|  H^{-1/2} K^* (1 + J\gamma_\Psi J^* )^{1/2})  \|_{\textrm{HS}} \\
&\leq \tr (H^{1/2} \widetilde{\alpha}_\Psi^* (1 + J\gamma_\Psi J^* )^{-1} \widetilde{\alpha}_\Psi H^{1/2} )^{1/2} \\
& \qquad 
\times \tr ( (1 + J\gamma_\Psi J^* )^{1/2}K  H^{-1} K^* (1 + J\gamma_\Psi J^* )^{1/2}) )^{1/2} \\
&\leq \left(\tr (H^{1/2} \gamma_\Psi H^{1/2})\right)^{1/2} \left(\tr(KH^{-1}K^*) + \tr(H^{1/2} \gamma_\psi H^{1/2}) \right)^{1/2} \\
&\leq \tr (H^{1/2} \gamma_\Psi H^{1/2}) + \tr(KH^{-1}K^*),
\end{align*}
where we used (\ref{eq:lem_Bog_ham}) and the assumption $KH^{-1}K^* \leq J H J^*$. Inserting this in (\ref{eq:lem_def_Ham}) concludes the proof.
\end{proof}

We now turn to another standard intermediate result.

\begin{lem}
	\label{lem:three_body}
Let $w \in L^{6/5}(\mathbb{R}^3)$, $f\in L^{\infty}(\mathbb{R}^3)$, $\chi_1,\chi_2 \in L^\infty(\mathbb{R})$. Then we have for all $\eta > 0$
\begin{multline*}
\pm \left( \chi_1 (\mathcal{N}) \iint f(x) w(x-y) a^*_y a_x a_y  \mathrm{d}x \, \mathrm{d}y \, \chi_2 (\mathcal{N}) + h.c. \right) 
\\ \leq \|w\|_{L^{6/5}(\mathbb{R}^3)} \|f\|_{L^\infty(\mathbb{R}^3)} \left( \eta \chi_1 (\mathcal{N})^2\mathcal{N}  + \eta^{-1} \chi_2 (\mathcal{N})^2 \mathcal{N} \dd \Gamma(- \Delta )\right).
\end{multline*}
\end{lem}

\begin{proof}
We recall the Cauchy-Schwarz inequality for operator:
\begin{equation*}
\pm(A B + B^* A^*) \leq \eta A A^* + \eta^{-1} B^* B, \qquad \forall \eta >0.
\end{equation*} 
We will use it for $A =  \chi_1 (\mathcal{N}) a^*_y$ and $B = f(x) w(x-y) a_x a_y  \chi_2 (\mathcal{N})$, we obtain for all $\eta >0$
\begin{align}
\pm \Big( &\iint   \chi_1 (\mathcal{N}) f(x) w(x-y) a^*_y a_x a_y  \mathrm{d}x \, \mathrm{d}y \,  \chi_2 (\mathcal{N})+ h.c. \Big) \nn \\
&\leq \eta \int  \chi_1 (\mathcal{N}) a^*_y a_y  \chi_1 (\mathcal{N})  \mathrm{d} y \, \nn \\
&\qquad 
	  +  \eta^{-1} \int \left( \iint w(x-y) f(x)  \overline{f(x')} \overline{w(x'-y)} \chi_2 (\mathcal{N}) a^*_x a^*_y a_{x'} a_y \chi_2 (\mathcal{N}) \mathrm{d} x \, \mathrm{d} x' \,\right) \mathrm{d} y \,.
	\label{eq:lem3body}
\end{align}
The second term above is $\chi_2 (\mathcal{N}) \mathbb{T} \chi_2 (\mathcal{N})$ where $\mathbb{T}$ is the second quantization of the operator $T$ defined by 
\begin{equation*}
T(\phi)(x,y) = w(x-y)f(x) \int \overline{w(x'-y) f(x')} \phi(x',y) \mathrm{d}x' \,,
\end{equation*}
for all $\phi \in H^{1}(\mathbb{R}^3 \times \mathbb{R}^3)$. It satisfies
\begin{align*}
\braket{\phi, T \phi}_{L^2(\mathbb{R}^3\times \mathbb{R}^3)} 
&= \int \left| \int \overline{ w(x-y) f(x)} \phi(x,y) \mathrm{d}x \, \right|^2 \mathrm{d}y \, \\
&\leq \int \| w \|_{L^{6/5}(\mathbb{R}^3)}^2  \|f\|_{L^\infty(\mathbb{R}^3)}^2 \|\phi(\cdot, y)\|_{L^6(\mathbb{R}^3)}^2 \mathrm{d}y \\
&\leq  \| w \|_{L^{6/5}(\mathbb{R}^3)}^2 \|f\|_{L^\infty(\mathbb{R}^3)}^2 \int \| \nabla_1 \phi(\cdot, y)\|_{L^2(\mathbb{R}^3)}^2 \mathrm{d}y \\
&\leq  \| w \|_{L^{6/5}(\mathbb{R}^3)}^2 \|f\|_{L^\infty(\mathbb{R}^3)}^2\braket{\phi, (-\Delta \otimes 1)-\phi}_{L^2(\mathbb{R}^3\times \mathbb{R}^3)}.
\end{align*}
Hence this yields the bound
\begin{equation*}
(\ref{eq:lem3body}) \leq \eta \mathcal{N}\chi_1 (\mathcal{N})^2 + C \eta^{-1} \| w \|_{L^{6/5}(\mathbb{R}^3)}^2 \|f\|_{L^\infty(\mathbb{R}^3)}^2 \mathcal{N} \chi_2 (\mathcal{N})^2 \dd \Gamma(-\Delta)
\end{equation*}
from which we obtain the desired result, using that $\mathcal{N} \leq \dd \Gamma(1-\Delta)$ and optimizing over $\eta$.
\end{proof}

%\begin{lem}
%	\label{lem:three_body_2}
%Let $w \in L^{2}(\mathbb{R}^3)$, $f\in L^{\infty}(\mathbb{R}^3)$, $\chi_1,\chi_2 \in L^\infty(\mathbb{R})$ then we have for all $\eta > 0$
%\begin{multline*}
%\pm \left( \chi_1 (\mathcal{N}) \iint f(x) w(x-y) a^*_y a_x a_y  \mathrm{d}x \, \mathrm{d}y \, \chi_2 (\mathcal{N}) + h.c. \right) 
%\\ \leq \|w\|_{L^{2}(\mathbb{R}^3)} \|f\|_{L^\infty(\mathbb{R}^3)} \left( \eta \chi_1 (\mathcal{N})^2\mathcal{N}  + \eta^{-1} \chi_2 (\mathcal{N})^2 \mathcal{N}^2 \right).
%\end{multline*}
%\end{lem}
%
%\begin{proof}
%The proof is the same as the one of \cref{lem:three_body} except that we use that the operator $T$ satisfies, for all $\phi \in L^2(\mathbb{R}^3 \times \mathbb{R}^3)$
%\begin{align*}
%\braket{\phi, T \phi}_{L^2(\mathbb{R}^3\times \mathbb{R}^3)} 
%&= \int \left| \int \overline{ w(x-y) f(x)} \phi(x,y) \mathrm{d}x \, \right|^2 \mathrm{d}y \, \\
%&\leq \int \| w \|_{L^{2}(\mathbb{R}^3)}^2  \|f\|_{L^\infty(\mathbb{R}^3)}^2 \|\phi(\cdot, y)\|_{L^2(\mathbb{R}^3)}^2 \mathrm{d}y \\
%&\leq  \| w \|_{L^{2}(\mathbb{R}^3)}^2 \|f\|_{L^\infty(\mathbb{R}^3)}^2 \|\phi\|_{L^2(\mathbb{R}^3\times\mathbb{R}^3)}.
%\end{align*}
%
%\end{proof}

\subsubsection{Bogoliubov's approximation: proof of \cref{lem:Bog_approx}}

We follow the proof \cite{NamNap-17d}. We emphasize that the inequalities (\ref{eq:lem_Bog_approx}) hold in $\mathcal{F}^{\leq m}(\mathcal{H}_+)$, that is 
\begin{equation*}
\langle \Phi,A \Phi \rangle \leq \mathds{1}^{\leq m} \leq C_{\varepsilon,t} \sqrt{\frac{m}{N^{1-\beta-\varepsilon}}} \langle \Phi, \dd \Gamma(1-\Delta)\Phi \rangle,  \quad \forall \Phi \in \mathcal{F}^{\leq m}(\mathcal{H}_+).
\end{equation*}
 Nevertheless, when possible, we will try to obtain first general estimates in $\mathcal{F}(\mathcal{H})$ and then take the projection on $\mathcal{F}^{\leq m}(\mathcal{H}_+)$. 

Let us begin by noting that $[R_0,\mathcal{N}] = [R_4,\mathcal{N}] = 0$, $[R_1,\mathcal{N}] = R_1$, $[R_2,\mathcal{N}] = -2 R_2$ and $[R_3,\mathcal{N}] = R_3$. It is therefore sufficient to prove (\ref{eq:lem_Bog_approx}) only for $A \in \{ \mathcal{E}_N(t), \partial_t \mathcal{E}_N(t) \}$.
From Hölder's inequality and the continuity property of the dipolar kernel $K$, see for instance \cite[Theorem 4.12]{Duo01}, for all $2\leq p<\infty$, there exists some constant $C_p >0$ such that for all $f\in L^p(\mathbb{R}^3)$ we have for all $N\geq 1$,
\begin{equation*}
\|w_N \ast f \|_{L^p(\mathbb{R}^3)} \leq C_p  \|f\|_{L^p(\mathbb{R}^3)}.
\end{equation*}
From this it follows that
\begin{align*}
\| \widetilde{K}_1(t) \|_{op} &\leq C_2  \|u_N(t)\|_{L^\infty(\mathbb{R}^3)}^2, \\
\| \partial \widetilde{K}_1(t) \|_{op} &\leq C_2 \|u_N(t)\|_{L^\infty(\mathbb{R}^3)} \|\partial_t u_N(t)\|_{L^\infty(\mathbb{R}^3)}.
\end{align*}
Similarly we have
\begin{equation*}
\|w_N \ast f \|_{L^\infty(\mathbb{R}^3)} \leq C_2 \|f\|_{H^2(\mathbb{R}^3)}.
\end{equation*}
Using this with the Hartree equation (\ref{eq:Hartree_bis}) we obtain $$\|\partial_t u_N(t)\|_{L^\infty(\mathbb{R}^3)} \leq C (1+C_2)\|u\|_{H^4(\mathbb{R}^3)}.$$
For the particular case $p=2$, we can take $C_2 = \|\widehat{w}\|_{L^\infty(\mathbb{R}^3)}$.
Finally, we recall that $\dd \Gamma(1) = \mathcal{N}$ and that for any $f\in L^2(\mathbb{R}^3)$ we have 
\begin{equation*}
a^*(f) a(f) \leq \|f\|_{L^2(\mathbb{R}^3)}^2 \mathcal{N}.
\end{equation*}
We will now pursue and estimate separately the terms involving $R_j$ for $j = 0 \dots 4$, where we recall that $\mathcal{E}_N(t) = \frac{1}{2} \sum_{j=0}^4 (R_j + R_j^*)$.
 \\

\subsection*{Step $1$: Bounds involving $R_0$}
We have
\begin{align*}
\pm R_0
&= \pm \dd \Gamma (Q(t) [w_N \ast |u_N(t)|^2 + \widetilde{K}_1(t) - \mu_N(t)]Q(t)) \frac{1-\mathcal{N}}{N-1} \\ 
& \leq C \frac{\mathcal{N}^2}{N} \Bigg( \| w_N \ast |u_N(t)|^2 \|_{L^{\infty}(\mathbb{R}^3)}  +  \|\widehat{w}\|_{L^\infty(\mathbb{R}^3)} \|u_N(t)\|_{L^\infty(\mathbb{R}^3)}^2  + \|\widehat{w}\|_{L^\infty(\mathbb{R}^3)} \|u_N(t)\|_{L^4(\mathbb{R}^3)}^4 \Bigg) \\
&\leq C_t \frac{\mathcal{N}^2}{N} ,
\end{align*}
which, after noting that $\mathcal{N} \leq \dd \Gamma(1-\Delta)$ and projecting on $\mathcal{F}^{\leq m}(\mathcal{H}_+)$, gives (\ref{eq:lem_Bog_approx}) for the $R_0$ part. We turn to the estimate of $\partial_t R_0$ and start by computing $$\partial_t Q(t) = - \ket{\partial_t u_N(t)}\bra{u_N(t)} - \ket{u_N(t)}\bra{\partial_t u_N(t)},$$ from which we have $$\|\partial_t Q(t)\|_{op} \leq 2 \| u_N(t)\|_{L^2(\mathbb{R}^3)} \|\partial_t u_N(t)\|_{L^2(\mathbb{R}^3)}.$$ Using the Cauchy-Schwarz inequality for operators we obtain
\begin{align*}
\pm \partial_t R_0 &= \pm \frac{1-\mathcal{N}}{1-N} \dd \Gamma \left( \partial_t Q(t) [w_N \ast |u_N(t)|^2 + \widetilde{K}_1(t) - \mu_N(t)] Q(t) + h.c. \right)  \\
& \quad+  \frac{1-\mathcal{N}}{1-N} \dd \Gamma \left( Q(t) [2 w_N \ast \Re (\partial_t u (t) \overline{u (t)}) + \partial_t \widetilde{K}_1(t) - \partial_t \mu_N(t)] Q(t) \right) \\
&\leq C \frac{\mathcal{N}^2}{N} \Big( \|\partial_t Q(t)\|_{op}^2 + \| w_N \ast |u_N(t)|^2\|_{L^\infty(\mathbb{R}^3)}^2 + \|\widehat{K}_1\|_{op}^2 \\
& \qquad \qquad \qquad \qquad \qquad 
	+ \|w_N \ast \Re(\partial_t u_N(t) \overline{u_N(t)})\|_{L^{\infty}(\mathbb{R}^3)} + \|\partial_t \widetilde{K}_1(t)\|_{op} + C_t \Big)  \\
&\leq C_t \frac{\mathcal{N}^2}{N}.
\end{align*}
Projecting on $\mathcal{F}^{\leq m}$ and noting that $\mathcal{N} \leq \dd \Gamma(1-\Delta)$ gives the result.

\subsection*{Step $2$: Bounds involving $R_1$}
Recall that for any $f \in L^2(\mathbb{R}^3)$ we have $a^*(f) a(f) \leq \|f\|_{L^2(\mathbb{R}^3)}^2 \mathcal{N}$. Hence, using the Cauchy-Schwarz inequality we obtain
\begin{align*}
\pm (R_1 + R_1^*) &= \mp 2 \left( \frac{\mathcal{N}\sqrt{N-\mathcal{N}}}{N-1} a(Q(t) [w_N\ast |u_N(t)|^2] u_N(t)) + h.c.\right) \\
&\leq C \eta \frac{\mathcal{N}^2}{N^{1/2}} + \eta^{-1} a^*(Q(t) [w_N\ast |u_N(t)|^2] u_N(t)) a(Q(t) [w_N\ast |u_N(t)|^2] u_N(t))  \\
&\leq C N^{-1/2}(\eta + \eta^{-1} \|[w_N\ast |u_N(t)|^2] u_N(t)\|_{L^2(\mathbb{R}^3)}^2\mathcal{N}) \mathcal{N} \\
&\leq C_t \frac{\eta + \eta^{-1}\mathcal{N}}{N^{1/2}} \mathcal{N}.
\end{align*}
Projecting onto $\mathcal{F}^{\leq m}$ and optimizing over $\eta$ gives the result. The term $\partial_t R_1$ is dealt with similarly.

\subsection*{Step $3$: Bounds involving $R_2$} 

Define $\chi(x) = 1 - \sqrt{(N-x)(N-x-1)}/(N-1)$ for $ x \leq N$ and note that $0 \leq \chi(x) \leq x/(N-1) $. Writing $Q(t) = 1 -\ket{u_N(t)}\bra{u_N(t)}$ in the expression of $R_2$ and expanding, we obtain after a simple computation that
\begin{multline}
	\label{eq:R_2}
R_2 = \Big( 2 a^*(u_N(t)) a^*([w_N \ast |u_N(t)|^2] u_N(t)) - 2 \mu_N(t) a^*(u_N(t)) a^*(u_N(t)) \\
	  - \iint u(t,x) w_N(x-y) u(t,y) a^*_x a^*_y \dd x \, \dd y \Big) \chi (\mathcal{N}).
\end{multline}
In the expression above, because of the $a^*(u_N(t))$ appearing in normal order, the first two terms vanish when the expectation is taken against an element of $\mathcal{F}(\mathcal{H}_+)$. We therefore focus on the last term. Applying \cref{lem:Bog_lower_bound} we have
\begin{align}
%	\label{eq:R_2_second_statement}
\pm \, &\left(\mathds{1}^{\leq m} \iint u(t,x) w_N(x-y) u(t,y) a^*_x a^*_y \dd x \, \dd y \, \chi (\mathcal{N}) + h.c. \right)\mathds{1}^{\leq m} \nn \\
&\qquad \qquad \qquad 
\leq C \left(\eta \|\mathds{1}^{\leq m} \chi(\mathcal{N})\|_{op}^2 \|(1-\Delta_x)^{-1/2} u_N(t) w_N\ast(u_N(t) \, \cdot)\|_{\mathfrak{S}_2}^2 + \eta^{-1} \dd \Gamma(1-\Delta)\right) \nn \\
&\qquad \qquad \qquad 
\leq C_{t,\varepsilon} \left(\eta \frac{m^2}{N^2} N^{\beta+\varepsilon} + \eta^{-1} \dd \Gamma(1-\Delta)\right) \nn \\
&\qquad \qquad \qquad 
\leq C_{t,\varepsilon} \sqrt{\frac{m}{N^{1-\beta-\varepsilon}}} \dd \Gamma(1-\Delta).\nn
\end{align}
We have used that 
\begin{align*}
\| (1-\Delta_x)^{-1/2} u_N(t) w_N\ast(u_N(t) \, \cdot)\|_{\mathfrak{S}_2} &\leq \| (1-\Delta_x)^{-1/2} u_N(t) \|_{\mathfrak{S}_{3+\varepsilon_3}} \| w_N\ast(u_N(t) \, \cdot)\|_{\mathfrak{S}_{6-\varepsilon_2}} \\
&\leq C \|u_N(t)\|_{L^{3+\varepsilon_3}} \|w_N\|_{L^{6/5+\varepsilon_1}(\mathbb{R}^3)} \|u_N(t)\|_{L^{6-\varepsilon_2}(\mathbb{R}^3)} \\
&\leq C_{t,\varepsilon} N^{\beta/2+ \varepsilon},
\end{align*}
where we used Hölder's inequality in Schatten spaces and the Kato-Seiler-Simon inequality: for any $p\geq 2$ and $f,g \in L^{p}(\mathbb{R}^3)$
\begin{equation*}
\|f(x) g(p)\|_{\mathfrak{S}_p} \leq C_p \|f\|_{L^p(\mathbb{R}^3)} \|g\|_{L^p(\mathbb{R}^3)},
\end{equation*}
for some constant $C_p >0$. We choose $\varepsilon_3,\varepsilon_2$ and $\varepsilon_1$ such that 
\begin{align*}
\frac{1}{2} &= \frac{1}{3+\varepsilon_3} + \frac{1}{6 - \varepsilon_2}, \quad 1 = \frac{1}{6-\varepsilon_2} + \frac{1}{6/5 + \varepsilon_1},
\quad 3\beta (1 - \frac{1}{6/5 + \varepsilon_1}) = \frac{\beta}{2} + \varepsilon.
\end{align*}
 We continue with the estimation of $\partial_t R_2$. Differentiating (\ref{eq:R_2}), we have
\begin{align*}
\partial_t R_2 &= 2\Big( a^*(\partial_t u_N(t)) a^*([w_N \ast |u_N(t)|^2] u_N(t)) + a^*( u_N(t)) a^*(\partial_t ([w_N \ast |u_N(t)|^2] u_N(t))) \\
&\qquad \qquad \qquad 
 -  \partial_t \mu_N(t) a^*(u_N(t)) a^*(u_N(t)) -  2 \mu_N(t) a^*(\partial_t u_N(t)) a^*(u_N(t))    \\
&\qquad \qquad \qquad \qquad \qquad \qquad \qquad 
  -  \iint \partial_t u(t,x) w_N(x-y) u(t,y) a^*_x a^*_y \dd x \, \dd y \Big) \chi (\mathcal{N}).
\end{align*}
Again, when taking the expectation with an element of $\mathcal{F}(\mathcal{H}_+)$, all the terms above containing $a^*(u_N(t))$ in normal order vanish. Hence, it remains to estimate
\begin{align*}
\pm & 2 \mathds{1}^{\leq m}\Big(  a^*(\partial_t u_N(t)) a^*([w_N \ast |u_N(t)|^2] u_N(t)) \\
&\qquad \qquad \qquad \qquad  \qquad 
-  \iint \partial_t u(t,x) w_N(x-y) u(t,y) a^*_x a^*_y \dd x \, \dd y  \Big) \chi (\mathcal{N})\mathds{1}^{\leq m} + h.c. \\
&\leq 4 \|\partial_t u_N(t)\|_{L^2(\mathbb{R}^3)} \|[w_N \ast |u_N(t)|^2] u_N(t)\|_{L^2(\mathbb{R}^3)} \chi (m) \mathcal{N} \\
& \qquad \qquad \qquad \qquad \qquad 
+ \left(\eta \frac{m^2}{N^2} \|(1-\Delta_x)^{-1/2} \partial_t u_N(t) w_N\ast(u_N(t) \, \cdot)\|_{\mathfrak{S}_2}^2 + \eta^{-1} \dd \Gamma(1-\Delta)\right) \\
&\leq C_{t,\varepsilon} \left( \frac{m}{N}\mathcal{N} +  \left( \eta \frac{m^2}{N^{2-\beta+\varepsilon}} + \eta^{-1} \dd \Gamma(1-\Delta)  \right)\right).
\end{align*}
Here we used again \cref{lem:Bog_lower_bound} and a similar argument as for estimating $\| (1-\Delta_x)^{-1/2} u_N(t) w_N\ast(u_N(t) \, \cdot)\|_{\mathfrak{S}_2}$. Projecting on $\mathcal{F}^{\leq m}$ and optimizing over $\eta$ yields the desired estimate.  \\

\subsection*{Step $4$: Bounds involving $R_3$}

Again, a computation shows that
\begin{multline}
	\label{eq:R_3}
R_3 = \frac{\sqrt{N-\mathcal{N}}}{N-1} \Bigg( \iiiint w_N(x-y) \overline{u(t,x)} a^*_y a_x a_y \dd x \, \dd y \, - \dd \Gamma (T(t)) a(u_N(t)) \\
\qquad - \dd \Gamma( [w_N\ast |u_N(t)|^2]) a(u_N(t)) + a^*(w\ast |u_N(t)|^2 u_N(t)) a(u_N(t)) a(u_N(t)) \\ + a^*(u_N(t)) a(u_N(t))a(w\ast |u_N(t)|^2 u_N(t))
- 2 \mu_N(t) a^*(u_N(t)) a(u_N(t))a(u_N(t)) \Bigg),
\end{multline}
where $T(t)$ is the operator defined by $T(t)(\phi) = w_N \ast(\overline{u_N(t)} \phi)$ for all $\phi \in L^2(\mathbb{R}^3)$. It is bounded with norm less than $\|\widehat{w}\|_{L^\infty(\mathbb{R}^3)} \|u_N(t)\|_{L^\infty(\mathbb{R}^3)}$. Dealing with $R_3$, and for the same reasons as previously, we only need to estimate the first term, which we do using \cref{lem:three_body}. We proceed as follows:
\begin{align*}
\pm \Big(\frac{\sqrt{N-\mathcal{N}}}{N-1} \iiiint w_N(x-y)& \overline{u(x)} a^*_y a_x a_y \dd x \, \dd y \,  + h.c. \Big) \\
&\leq C N^{-1/2} \|w_N\|_{L^{6/5}(\mathbb{R}^3)} \|u_N(t)\|_{L^\infty(\mathbb{R}^3)} \left( \eta \mathcal{N} + \eta^{-1} \right) \dd \Gamma(1-\Delta) \\
&\leq C_t \frac{\eta \mathcal{N} + \eta^{-1}}{N^{(1-\beta)/2}}  \dd \Gamma(1-\Delta),
\end{align*}
for all $\eta >0$. Projecting on $\mathcal{F}^{\leq m}$ and optimizing over $\eta$ gives the result. We now continue with the estimates involving $\partial_t R_3$. Differentiating (\ref{eq:R_3}) we have
\begin{multline*}
\partial_t R_3 = \frac{\sqrt{N-\mathcal{N}}}{N-1} \Bigg( \iiiint w_N(x-y) \overline{\partial_t u(t,x)} a^*_y a_x a_y \dd x \, \dd y \, - \dd \Gamma (\partial_t T(t)) a(u_N(t)) \\
 - \dd \Gamma (T(t)) a(\partial_t u_N(t)) - \dd \Gamma(\partial_t ( [w_N\ast |u_N(t)|^2])) a(u_N(t)) - \dd \Gamma( [w_N\ast |u_N(t)|^2]) a(\partial_t u_N(t)) \\
 + a^*(\partial_t(w\ast |u_N(t)|^2 u_N(t))) a(u_N(t))a(u_N(t))  + 2 a^*(w\ast |u_N(t)|^2 u_N(t)) a(\partial_t u_N(t))a(u_N(t)) \\ + a^*(\partial_t u_N(t)) a(u_N(t))a(w\ast |u_N(t)|^2 u_N(t)) + a^*(u_N(t)) a(\partial_t u_N(t))a(w\ast |u_N(t)|^2 u_N(t)) \\ + a^*(u_N(t)) a(u_N(t))a( \partial_t (w\ast |u_N(t)|^2 u_N(t)))
 - 2 \partial_t\mu_N(t) a^*(u_N(t)) a(u_N(t))a(u_N(t)) \\- 2 \mu_N(t) a^*(\partial_t u_N(t)) a(u_N(t))a(u_N(t)) - 4 \mu_N(t) a^*(u_N(t)) a(\partial_t u_N(t))a(u_N(t))\Bigg).
\end{multline*}
Again, any term containing $a^*(u_N(t))$ or $a(u_N(t))$ in normal order vanishes when taking the expectation with an element of $\mathcal{F}(\mathcal{H}_+)$. It remains to estimate
\begin{align*}
\pm& \frac{\sqrt{N-\mathcal{N}}}{N-1} \Bigg( \iiiint w_N(x-y) \overline{\partial_t u(t,x)} a^*_y a_x a_y \dd x \, \dd y \, - \dd \Gamma (T(t)) a(\partial_t u_N(t)) \\
&\qquad \qquad \qquad \qquad \qquad \qquad \qquad \qquad \qquad \qquad \qquad 
- \dd \Gamma( [w_N\ast |u_N(t)|^2]) a(\partial_t u_N(t))  + h.c. \Bigg)  \\
&\leq C N^{-1/2} \Bigg( \|w_N\|_{L^{6/5}(\mathbb{R}^3)} \|\partial_t u_N(t)\|_{L^\infty(\mathbb{R}^3)} \left( \eta \mathcal{N} + \eta^{-1} \right) \dd \Gamma(1-\Delta) + \eta' a^*(\partial_t u_N(t))a(\partial_t u_N(t)) \\
& \qquad \qquad \qquad \qquad \qquad \qquad \qquad \qquad  
+ (\eta')^{-1} \Big(\dd \Gamma (T(t)) \dd \Gamma (T(t)^*) + \dd \Gamma( [w_N\ast |u_N(t)|^2])^2 \Big)  \Bigg) \\
&\leq C_t N^{-1/2} \Big( N^{\beta/2}\left( \eta \mathcal{N} + \eta^{-1} \right) \dd \Gamma(1-\Delta) + \eta' \|\partial_t u_N(t)\|_{L^2(\mathbb{R}^3)}^2 \mathcal{N} \\
& \qquad \qquad \qquad \qquad \qquad \qquad \qquad \qquad \qquad 
+ (\eta')^{-1} \Big( \|T(t)\|_{op}^2 + \|w_N \ast |u_N(t)|^2\|_{L^\infty(\mathbb{R}^3)}^2 \Big) \mathcal{N}^2\Big) \\
 &\leq C_t N^{(\beta-1)/2}  \left( \eta \mathcal{N} + \eta^{-1} \right) \dd \Gamma(1-\Delta),
\end{align*}
for all $\eta > 0$. We used the Cauchy-Schwarz inequality and \cref{lem:three_body} to obtain the second inequality. Projecting on $\mathcal{F}^{\leq m}(\mathcal{H}_+)$ and optimizing with respect to $\eta>0$, we obtain the desired result.\\

\subsection*{Step $5$: Bounds involving $R_4$}

From $\pm w_N(x-y) \leq N^{\beta} \|w\|_{L^{3/2}(\mathbb{R}^3)} (1-\Delta_x)$ one has
\begin{align*}
\pm R_4 
&\leq C N^{\beta-1} \dd \Gamma(Q(t)(1-\Delta_x)Q(t)) \dd \Gamma(Q(t)) \\
&\leq C \frac{\mathcal{N}}{N^{1-\beta}} \dd \Gamma(Q(t)(1-\Delta_x)Q(t)).
\end{align*}
Next we turn to
\begin{align*}
\partial_t R_4 
&= \frac{1}{(N-1)} \iiiint (\partial_t Q(t)\otimes Q(t) w_N Q(t) \otimes Q(t))(x,y,x',y') a^*_x a^*_y a_{x'} a_{y'} \dd x \,\dd y \,\dd x' \, \dd y' \, + h.c. \\
&= - \frac{1}{(N-1)}  \iiint (1 \otimes Q(t) w_N Q(t)\otimes Q(t)) (x,y,x',y') \times \\
&\qquad \qquad \qquad \qquad  
 \times \left( \overline{\partial_t u(t,x)} a^*(u_N(t)) +  \overline{u(t,x)} a^*(\partial_ t u_N(t))  \right)  a^*_y a_{x'} a_{y'} \dd x \,\dd y \,\dd x' \, \dd y'  + h.c.
\end{align*}
Again, since we are interested in taking the expectation of an element of $\mathcal{F}(\mathcal{H}_+)$, we can ignore the terms containing $a^*(u_N(t))$ or $a(u_N(t))$ and consider the remaining terms where $Q(t)$ is replaced by $1$. Then the same computations as in the proof of \cref{lem:three_body} but replacing $\chi_1$ by $a^*(\partial_t u_N(t))$ give
\begin{align*}
\pm &\frac{1}{N-1} \left( \iint w_N(x-y) u(t,x) a^*(\partial_t u_N(t)) a^*_y a_{y} a_x \, \dd x \,\dd y  + h.c.\right) \\
&\leq  \frac{1}{N-1} \left( \eta \|\partial_t u_N(t)\|_{L^2(\mathbb{R}^3)}^2 \mathcal{N}^2 + \eta^{-1} N^{\beta} \|w\|_{L^{6/5}(\mathbb{R}^3)}^2 \|u_N(t)\|_{L^\infty(\mathbb{R}^3)}^2 \mathcal{N} d\Gamma(1 - \Delta)\right) \\
&\leq C_t \frac{\mathcal{N}}{N^{1-\beta/2}}  d\Gamma(1 - \Delta),
\end{align*}
where we have optimized over $\eta >0$. Projecting on $\mathcal{F}^{\leq m}$ concludes the proof of Lemme \ref{lem:Bog_approx}.
%, it remains to estimate
%\begin{align*}
%\pm &\frac{1}{N-1} \left( \iint w_N(x-y) u(t,x) a^*(\partial_t u_N(t)) a^*_y a_{y} \, \dd x \,\dd y  + h.c.\right) \\
%&\leq \frac{1}{N-1} \left( \eta a^*(\partial_t u_N(t)) \int  a^*_y a_{y} \,\dd y a(\partial_t u_N(t)) + \eta^{-1} \iiint w(x-y) u(t,x) w(z-y) \overline{u(t,z)}  \dd x \,\dd y  \dd z \, \right) \\
%&\leq \frac{1}{N-1} \left( \eta \|\partial_t u_N(t)\|_{L^2(\mathbb{R}^3)}^2 \mathcal{N}^2 + \eta^{-1} N^{2 \beta} \|w\|_{L^{6/5}(\mathbb{R}^3)}^2 \|u_N(t)\|_{L^\infty(\mathbb{R}^3)}^2 \mathcal{N} d\Gamma(1 - \Delta ).
%\end{align*}
\qed
\subsubsection{Bogoliubov stability: proof of \cref{lem:stability_H}}

Recall that
\begin{multline*}
\mathbb{H}(t) + \dd \Gamma(\Delta) = \dd \Gamma(1 + w\ast |u_N(t)|^{2} + Q(t) \widetilde{K}_{1}(t) Q(t) - \mu_N(t)) \\
+ \left( \iint K_2(t,x,y) a^*_x a^*_y \dd x \, \dd y \, + \iint \overline{K_2(t,x,y)} a_x a_y \dd x \, \dd y \right).
\end{multline*}
For the first term we have
\begin{multline*}
\pm \dd \Gamma(1 + w_N\ast |u_N(t)|^{2} + Q(t) \widetilde{K}_{1}(t) Q(t) - \mu_N(t)) \\
\leq \left(1 + \|w_N\ast |u_N(t)|^{2} \|_{L^\infty(\mathbb{R}^3)} + \|\widetilde{K}_{1}(t)\|_{op} + |\mu_N(t)| \right) \mathcal{N}.
\end{multline*}
We expand the second term and we use \cref{lem:Bog_lower_bound},
\begin{align}
	\label{eq:decompo_K2}
\pm \Big( \iint &K_2(t,x,y) a^*_x a^*_y \dd x \, \dd y \, + h.c. \Big) \nn\\
&= \pm \Bigg( \iint w_N (x-y)u(t,x)u(t,y) a^*_x a^*_y \dd x \, \dd y \,\nn \\ 
& \qquad 
- 2 a^*([w_N\ast |u_N(t)|^2] u_N(t)) a^*(u_N(t)) + 2 \mu_N(t) a^*(u_N(t)) a^*(u_N(t)) + h.c. \Bigg) \\
&\leq C\Big(  \eta \| (1-\Delta_x)^{-1/2} u_N(t) w_N\ast(u_N(t) \, \cdot)\|_{\mathfrak{S}_2}^2 + \eta^{-1} \dd \Gamma(1-\Delta) \nn \\
&\qquad \qquad 
+ \|[w_N \ast |u_N(t)|^2]u_N(t)\|_{L^2(\mathbb{R}^3)}\|u_N(t)\|_{L^2(\mathbb{R}^3)} + \|u_N(t)\|_{L^2(\mathbb{R}^3)}^2 |\mu_N(t)| \mathcal{N} \Big)\nn \\
&\leq C_{t,\varepsilon} \left( \mathcal{N} + \eta N^{\beta+\varepsilon} + \eta^{-1} \dd \Gamma(1-\Delta) \right), \nn
\end{align}
for $\eta>0$. We then evaluate
\begin{multline*}
\partial_t \mathbb{H}(t) = \dd \Gamma\left( 2 w_N \ast \Re (\overline{\partial_t u_N(t)} u_N(t)) + \partial_t (Q(t)\widetilde{K}_1(t) Q(t) ) - \partial_t \mu_N(t) \right) \\
+ \iint \partial_t K_2(t,x,t) a^*_x a^*_y \dd x \, \dd y \, + \iint \overline{ \partial_t K_2(t,x,t)} a_x a_y \dd x \, \dd y \,.
\end{multline*}
For the first term we have
\begin{align*}
\pm \dd \Gamma\Big( w_N \ast \Re (\overline{\partial_t u_N(t)} u_N(t)) + &\partial_t (Q(t)\widetilde{K}_1(t) Q(t) ) - \partial_t \mu_N(t) \Big)  \\
&\leq \Bigg(\|w_N \ast \Re (\overline{\partial_t u_N(t)} u_N(t)) \|_{L^\infty(\mathbb{R}^3)} + \|\partial_t\widetilde{K}_1(t)\|_{op}  \\
& \qquad \qquad \qquad 
	 +  \|\widetilde{K}_1(t)\|_{op} \|\partial_t u_N(t)\|_{L^2(\mathbb{R}^3)}+ |\partial_t \mu_N(t)| + 1\Bigg) \mathcal{N} \\
&\leq C_t \mathcal{N}.
\end{align*}
To estimate the second and third terms, we differentiate (\ref{eq:decompo_K2}) and obtain that
\begin{align*}
\pm \Big( &\iint \partial_t K_2(t,x,y) a^*_x a^*_y \dd x \, \dd y \, + h.c. \Big) \nn\\
&= \pm \Bigg( 2 \iint w_N (x-y)(\partial_t u(t,x))u(t,y) a^*_x a^*_y \dd x \, \dd y \,- 2 a^*(\partial_t ([w_N\ast |u_N(t)|^2] u_N(t))) a^*(u_N(t))\nn \\ 
& \qquad \qquad \qquad 
+ a^*([w_N\ast |u_N(t)|^2] u_N(t)) a^*(\partial_t u_N(t)) + 2 \partial_t \mu_N(t) a^*(u_N(t)) a^*(u_N(t)) \\
& \qquad \qquad \qquad \qquad \qquad \qquad \qquad \qquad \qquad \qquad \qquad 
 + 2 \mu_N(t) a^*(\partial_t u_N(t)) a^*(u_N(t))  + h.c. \Bigg) \\
& \leq C \Bigg( \eta \| (1-\Delta_x)^{-1/2} \partial_t u_N(t) w_N\ast(u_N(t) \, \cdot)\|_{\mathfrak{S}_2}^2 + \eta^{-1} \dd \Gamma(1-\Delta) \\
& \qquad \qquad  + \Big(\|\partial_t ([w_N\ast |u_N(t)|^2] u_N(t))\|_{L^2(\mathbb{R}^3)} \|u_N(t)\|_{L^2(\mathbb{R}^3)} \\
	&\qquad \qquad \qquad \qquad 
  +  \| ([w_N\ast |u_N(t)|^2] u_N(t))\|_{L^2(\mathbb{R}^3)} \| \partial_t u_N(t)\|_{L^2(\mathbb{R}^3)} \\
  & \qquad \qquad \qquad \qquad
  + |\partial_t \mu_N(t)| \| u_N(t)\|_{L^2(\mathbb{R}^3)}^2 
 + \mu_N(t)  \| u_N(t)\|_{L^2(\mathbb{R}^3)} \| \partial_t u_N(t)\|_{L^2(\mathbb{R}^3)}  \Big)  \mathcal{N} \Bigg) \\
 &\leq C_{t,\varepsilon} \left(\mathcal{N} + \eta N^{\beta + \varepsilon} + \eta^{-1} \dd \Gamma(1-\Delta)\right),
\end{align*}
where we used the Cauchy-Schwarz inequality and \cref{lem:Bog_lower_bound}.

Finally, since
\begin{equation*}
i[\mathbb{H},\mathcal{N}] = -\iint \left( i K_2(t,x,y) a^*_x a^*_y  + \overline{ i K_2(t,x,y)} a_x a_y \right) \dd x \, \dd y \,,
\end{equation*}
we can estimate this term in a similar manner as before and obtain the desired bound. \qed

\subsection{Norm approximation}

We follow and adapt the arguments in \cite{NamNap-17d}, we obtain the following lemma.
\begin{lem}
	\label{lem:truncated_dyn}
Let $M= N^{1-\delta}$ with $\delta \in (0,1)$, then we have
\begin{equation*}
\|\Phi_N(t) - \Phi_{N,M}(t)\|_{L^2(\mathbb{R}^{3N})}^2 \leq C_{t,\varepsilon} \left(\frac{1}{M^{1/2}} + (|e_N| + N)^{1/4}\left( \frac{N^{3(\beta+\varepsilon)/4}}{M}+ \frac{N^{(\beta + \varepsilon-1)/2}}{M^{1/4}}\right)\right).
\end{equation*}
\end{lem}

\begin{proof}
The proof follows the one of \cite{NamNap-17d}, it differs in that it uses \cref{lem:three_body} to deal with three body terms and that one has to be a little bit more careful when estimating the two-body terms.

We have
\begin{equation*}
\|\Phi_N(t) - \Phi_{N,M}(t)\|_{L^2(\mathbb{R}^{3N})}^2 = 2\left(1 - \Re \langle \Phi_{N}(t), \Phi_{N,M}(t) \rangle \right).
\end{equation*}
Let $M/2 \leq m \leq M-3$ and decompose
\begin{equation*}
\langle \Phi_{N}(t), \Phi_{N,M}(t) \rangle = \langle \Phi_{N}(t),\mathds{1}^{\leq m} \Phi_{N,M}(t) \rangle +  \langle \Phi_{N}(t),\mathds{1}^{> m} \Phi_{N,M}(t) \rangle.
\end{equation*}
The second term is estimated using the Cauchy-Schwarz inequality
\begin{align}
	\label{eq:truncated_dyn_1}
|  \langle \Phi_{N}(t),\mathds{1}^{> m} \Phi_{N,M}(t) \rangle | 
& \leq \| \Phi_{N}(t)\|_{L^2(\mathbb{R}^{3N})} \|\mathds{1}^{> m} \Phi_{N,M}(t)\|_{L^2(\mathbb{R}^{3N})} \nn \\
&\leq \langle \Phi_{N,M}(t), (\mathcal{N}/m) \Phi_{N,M}(t) \rangle^{1/2} \nn \\
&\leq C_t M^{-1/2}.
\end{align}
We now want to prove that the first term remains close to $1$. To this aim we compute its time derivative
\begin{align*}
\frac{\dd }{\dd t} \langle \Phi_{N}(t),\mathds{1}^{\leq m} \Phi_{N,M}(t) \rangle  = \langle \Phi_{N}(t), i [\mathcal{G}_N(t),\mathds{1}^{\leq m}] \Phi_{N,M}(t) \rangle 
\end{align*}
and consider its average over the parameter $M/2 \leq m \leq M-3$
\begin{equation*}
\frac{1}{M/2 - 2} \sum_{m=M/2}^{M-3} \langle \Phi_{N}(t), i [\mathcal{G}_N(t),\mathds{1}^{\leq m}] \Phi_{N,M}(t) \rangle.
\end{equation*}
The gain obtained by averaging comes from the fact that the commutator $ [\mathcal{G}_N(t),\mathds{1}^{\leq m}] $ is localized in $\{ m-2 \leq \mathcal{N} \leq m+2\}$. As was shown in \cite{NamNap-17d} we have
\begin{equation*}
\sum_{m=M/2}^{M-3}  i [\mathcal{G}_N(t),\mathds{1}^{\leq m}] =  A_1 \chi_1(\mathcal{N})^2 + A_2 \chi_{2}(\mathcal{N})^2 + h.c. \, ,
\end{equation*}
% We write $ [\mathcal{G}_N(t),\mathds{1}^{\leq m}] = \mathds{1}^{> m} \mathcal{G}_N(t)\mathds{1}^{\leq m} - \mathds{1}^{\leq m}\mathcal{G}_N(t) \mathds{1}^{> m}  $ and focus on the first term. The second one is dealt with similarly. We have
%\begin{equation*}
% \sum_{m=M/2}^{M-3} \mathds{1}^{> m} \mathcal{G}_N(t)\mathds{1}^{\leq m} = A_1 \chi_1(\mathcal{N})^2 + A_2 \chi_{2}(\mathcal{N})^2
%\end{equation*}
where
\begin{align*}
A_1 &= \frac{i}{2} \iiiint (Q(t) \otimes Q(t) w_N Q(t) \otimes 1)(x,y,x',y') u(t,x) a^*_x a^*_{y} a_{y'}  \dd x \, \dd y \,\dd x' \,\dd y' \, \\
&\qquad \qquad \qquad \qquad \qquad 
- a^*(Q(t) [w\ast |u_N(t)u_N(t)|^2] u_N(t)) \mathcal{N} \\
& =:A_1^3 + A_1^1, \\
A_2 &= \frac{i}{2} \iint K_2(t,x,y)  a^*_x  a^*_y \dd x \, \dd y \,,
\end{align*}
and
\begin{align*}
\chi_1(\mathcal{N})^2 &=  \frac{\sqrt{N-\mathcal{N}}}{N-1} \mathds{1}(M/2 \leq \mathcal{N} \leq M-3), \\
\chi_2(\mathcal{N})^2 &= \frac{\sqrt{(N-\mathcal{N})(N-\mathcal{N}-1)}}{N-1}[\mathds{1}(M/2 -1 < \mathcal{N} \leq M-3) + \mathds{1}(M/2 \leq \mathcal{N} < M-3)].
\end{align*}
Note that since $\Phi_{N}(t), \Phi_{N,M}(t) \in \mathcal{F}(\mathcal{H}_+)$, we can replace $Q(t)$ by $1$ in the expression of the quantities $\langle \Phi_N(t), A_i^j \chi_i(\mathcal{N})\Phi_{N,M}(t)\rangle$. We have
\begin{align*}
\Bigg| \langle \Phi_N(t) , &\Big( A_1^3 \chi_1(\mathcal{N})^2 + h.c. \Big) \Phi_{N,M}(t) \rangle \Bigg|  \\
&= \frac{1}{2} \left|  \langle  \Phi_N(t), \left(\int iw_N(x - y)u(t,x) \chi_1(\mathcal{N}-1)a^*_x a^*_y a_y \chi_1(\mathcal{N})  \dd x \, \dd y + h.c. \right)\Phi_{N,M}(t) \rangle \right| \\
&\leq   CN^{-1/2} \langle \Phi_N(t), \|w_N\|_{L^{6/5}(\mathbb{R}^3)} \mathds{1}^{\geq M/2+1} \mathds{1}^{\leq M-2} \left(\eta_1 \mathcal{N} + \eta_1^{-1}\mathcal{N}\dd \Gamma(1-\Delta)\right)\Phi_N(t)\rangle^{\tfrac{1}{2}} \\
&\qquad 
	  \times \langle \Phi_{N,M}(t), \|w_N\|_{L^{6/5}(\mathbb{R}^3)} \mathds{1}^{\geq M/2+1} \mathds{1}^{\leq M-2}\left(\eta_2 \mathcal{N} + \eta_2^{-1}\mathcal{N} \dd \Gamma(1-\Delta) \right)\Phi_{N,M}(t)\rangle^{\tfrac{1}{2}} \\
&\leq C_{t,\varepsilon} N^{-1/2} \|w \|_{L^{6/5}(\mathbb{R}^3)}  \Big\{ N^{\beta/2} \left(\eta_1 M +\eta_1^{-1}M (|e_N|+N) \right) \times \\
& \qquad \qquad \qquad \qquad \qquad \qquad \qquad \qquad \qquad 
	\times N^{\beta/2} \left(\eta_2 N^{\beta+\varepsilon} +\eta_2^{-1}M N^{\beta+\varepsilon}\right)\Big\}^{\tfrac{1}{2}} \\
& \leq C_{t,\varepsilon} N^{-1/2} N^{\beta + \varepsilon/2} M^{3/4} (|e_N|+N)^{1/4}.
\end{align*} 
We have used that $$\langle \Phi_{N,M}(t), \mathcal{N} \Phi_{N,M}(t) \rangle \leq \langle \Phi_{N,M}(t), \dd \Gamma (1-\Delta) \Phi_{N,M}(t) \rangle \leq C_{t,\varepsilon} N^{\beta + \varepsilon}$$ and the estimate (\ref{eq:bound_kinetic_true_state}). Next we have
\begin{align*}
\Bigg| \langle &\Phi_N(t) , A_1^1 \chi_1(\mathcal{N})^2 \Phi_{N,M}(t) \rangle \Bigg| \\
&= \Bigg| \langle \Phi_N(t) , \chi_1(\mathcal{N}-1) a^*(Q(t) [w\ast |u_N(t)|^2] u_N(t)) \mathcal{N} \chi_1(\mathcal{N}) \Phi_{N,M}(t) \rangle \Bigg| \\
&\leq CN^{-1/2} \langle \Phi_N(t) , \mathds{1}^{\leq M+1} a^*(Q(t) [w\ast |u_N(t)|^2] u_N(t)) a(Q(t) [w\ast |u_N(t)|^2] u_N(t))  \Phi_N(t)\rangle^{1/2} \times \\
&\quad \times \langle \Phi_{N,M}(t) ,\mathds{1}^{\leq M} \mathcal{N}^2 \Phi_{N,M}(t)\rangle^{1/2} \\
&\leq  C \|[w_N \ast |u_N(t)|^2] u_N(t)\|_{L^2(\mathbb{R}^3)} N^{-1/2} \langle \Phi_N(t) , \mathds{1}^{\leq M} \mathcal{N}  \Phi_N(t)\rangle^{1/2} 
\langle \Phi_{N,M}(t) ,\mathds{1}^{\leq M} \mathcal{N}^2 \Phi_{N,M}(t)\rangle^{1/2} \\
&\leq C_{t,\varepsilon} N^{-1/2} N^{(\beta+\varepsilon)/2} M.
\end{align*}
The term with $\langle \Phi_N(t) ,  \chi_1(\mathcal{N})^2 (A_1^1)^* \Phi_{N,M}(t) \rangle$ is dealt with similarly. Finally, we apply \cref{lem:Bog_lower_bound} as well as the Cauchy-Schwarz inequality to bound the last term,
\begin{align}
	\label{eq:truncated_dyn_11}
\Bigg| &\langle \Phi_N(t) , (A_2 \chi_2(\mathcal{N})^2 +h.c. ) \Phi_{N,M}(t) \rangle \Bigg| \nn \\
&= \Bigg| \langle \Phi_N(t) , \left(\iint iw_N(x-y) u(x) u(y) a^*_x a^*_y \chi_2(\mathcal{N})^2 \dd x \,\dd y + h.c.\right) \Phi_{N,M}(t) \rangle \Bigg|   \nn\\
&\leq  \left( \eta \|\chi_2(\mathcal{N})\|_{op}^2 \|(1-\Delta_x)^{-1/2} u_N(t) w_N \ast(u_N(t) \cdot)\|_{\mathfrak{S}_2}^2 + \eta^{-1} \langle \Phi_N(t)\dd \Gamma(1-\Delta) \Phi_N(t) \rangle\right)^{1/2}  \times \nn\\
&\quad \times  \left( \eta' \|\chi_2(\mathcal{N})\|_{op}^2 \|(1-\Delta_x)^{-1/2} u_N(t) w_N \ast(u_N(t) \cdot)\|_{\mathfrak{S}_2}^2 + (\eta')^{-1}\langle \Phi_{N,M}(t)\dd \Gamma(1-\Delta) \Phi_{N,M}(t) \rangle\right)^{1/2} \nn\\
&\leq \left(\eta  N^{\beta+\varepsilon} + \eta^{-1} \langle \Phi_N(t), \dd \Gamma(1-\Delta)\Phi_N(t) \rangle  \right)^{1/2} \times\nn \\
&\qquad \qquad \qquad \qquad\qquad
\times \left(\eta' N^{\beta+\varepsilon} + (\eta')^{-1} \langle \Phi_{N,M}(t),\dd \Gamma(1-\Delta)\Phi_{N,M}(t) \rangle  \right)^{1/2}.
\end{align}
Now we use again that $\langle \Phi_{N,M}(t), \dd \Gamma(1-\Delta)\Phi_{N,M}(t) \rangle \leq C_{t,\varepsilon} N^{\beta +\varepsilon}$ and that $\langle \Phi_{N}(t), \dd \Gamma(1-\Delta)\Phi_{N}(t) \rangle \leq C (|e_N| + N)$. After optimizing over $\eta$ and $\eta'$ we obtain
\begin{equation*}
(\ref{eq:truncated_dyn_11}) \leq C_{t,\varepsilon} \frac{(|e_N| + N)^{1/4}N^{3(\beta+\varepsilon)/4}}{M}.
\end{equation*}
Hence we have shown that
\begin{multline}
		\label{eq:truncated_dyn_2}
\Bigg | \frac{\dd }{\dd t} \left(\frac{1}{M/2 - 2} \sum_{m=M/2}^{M-3}  \langle \Phi_{N}(t),\mathds{1}^{\leq m} \Phi_{N,M}(t) \rangle\right) \Bigg | \\
\leq C_{t,\varepsilon} \left(\frac{(|e_N| + N)^{1/4}N^{3(\beta+\varepsilon)/4}}{M} +  \frac{(|e_N|+N)^{1/4} N^{(\beta + \varepsilon-1)/2}}{M^{1/4}} + N^{(\beta+\varepsilon - 1)/2}\right).
\end{multline}
On the other hand, recall that $\Phi_{N,M}(0) = \Phi_{N}(0) = \Phi(0)$, so that for $M/2 \leq m \leq M-3$,
\begin{align}
		\label{eq:truncated_dyn_3}
 \langle \Phi_{N}(0),\mathds{1}^{\leq m} \Phi_{N,M}(0)  \rangle &=\langle  \Phi(0),\mathds{1}^{\leq m} \Phi (0) \rangle \nn \\
& = 1 - \langle  \Phi(0),\mathds{1}^{> m} \Phi (0) \rangle \nn \\
&\geq 1 - \langle  \Phi(0),\mathds{1}^{> m} (\mathcal{N}/m) \Phi (0) \rangle  \nn \\
&\geq 1 - C M^{-1}.
\end{align}
Gathering (\ref{eq:truncated_dyn_1}), (\ref{eq:truncated_dyn_2}) (\ref{eq:truncated_dyn_3}) we obtain
\begin{multline*}
\|\Phi_N(t) - \Phi_{N,M}(t)\|_{L^2(\mathbb{R}^{3N})}^2 \\ \leq C_{t,\varepsilon} \left(\frac{1}{M^{1/2}} + \frac{(|e_N| + N)^{1/4}N^{3(\beta+\varepsilon)/4}}{M} + \frac{(|e_N|+N)^{1/4} N^{(\beta + \varepsilon-1)/2}}{M^{1/4}} \right).
\end{multline*}
\end{proof}

As in \cite{NamNap-17d} we compare the Bogoliubov dynamics and the truncated one.
\begin{lem}
	\label{lem:bogo_trunc_dyn}
Let $M= N^{\alpha}$ with $0 < \alpha < 1 - \beta$, let $N$ be large enough, then we have
\begin{equation*}
\|\Phi(t) - \Phi_{N,M}(t)\|_{L^2(\mathbb{R}^{3N})}^2 \leq C_{t,\varepsilon} \left(\frac{1}{M^{1/2}}  + \frac{N^{\beta+\varepsilon}}{M} + \frac{M}{N^{1-2\beta - \varepsilon}} + \frac{M}{N^{(1-2\beta - 2 \varepsilon)/2}}\right).
\end{equation*}
\end{lem}

\begin{proof}
The proof is similar as the one in \cite{NamNap-17d} except that we use the estimates of \cref{lem:Bog_approx}. As before, we have
\begin{equation}
	\label{eq:Bog_trunc_dyn_0}
\|\Phi (t) - \Phi_{N,M}(t)\|_{L^2(\mathbb{R}^{3N})}^2 \leq 2\left(1 - \Re \langle \Phi_{N,M} (t), \Phi(t) \rangle \right).
\end{equation}
We let $M/2 \leq m \leq M-3$ and decompose
\begin{equation*}
\langle \Phi_{N,M} (t), \Phi (t) \rangle = \langle \Phi_{N,M} (t),\mathds{1}^{\leq m} \Phi (t) \rangle +  \langle \Phi_{N,M}(t),\mathds{1}^{> m} \Phi (t) \rangle.
\end{equation*}
The second term is bounded by the Cauchy-Schwarz inequality
\begin{align}
	\label{eq:Bog_trunc_dyn_1}
|  \langle \Phi_{N,M} (t),\mathds{1}^{> m} \Phi (t) \rangle | 
& \leq \|\mathds{1}^{> m} \Phi_{N,M} (t)\|  \|\mathds{1}^{> m} \Phi (t)\| \nn \\
&\leq \langle \Phi_{N,M} (t), (\mathcal{N}/m) \Phi_{N,M}(t) \rangle^{1/2}  \times \langle \Phi (t), (\mathcal{N}/m) \Phi (t) \rangle^{1/2} \nn \\
&\leq C_t M^{-1}.
\end{align}
As in the proof of \cref{lem:truncated_dyn}, we will show that the first term remains close to $1$, we compute its time derivative
\begin{align*}
\frac{\dd }{\dd t} \langle \Phi_{N,M}(t),\mathds{1}^{\leq m} \Phi(t) \rangle  
= i \langle \Phi_{N,M}(t), \left((\mathcal{G}_N(t) - \mathbb{H})\mathds{1}^{\leq m} +  i [\mathbb{H},\mathds{1}^{\leq m}]\right) \Phi(t) \rangle.
\end{align*}
The first term is estimated using \cref{lem:Bog_approx} and the Cauchy-Schwarz inequality, we obtain
\begin{align}
		\label{eq:Bog_trunc_dyn_2}
\Bigg| \langle &\Phi_{N,M}(t),  (\mathcal{G}_N(t) - \mathbb{H})\mathds{1}^{\leq m} \Phi(t) \rangle \Bigg| \nn \\
&= \Bigg| \langle \Phi_{N,M}(t), \mathds{1}^{\leq M} (\mathcal{G}_N(t) - \mathbb{H}) \mathds{1}^{\leq M} \mathds{1}^{\leq m} \Phi (t) \rangle \Bigg| \nn \\
&\leq C_{t,\varepsilon} \Big\langle \Phi(t), \sqrt{\frac{M}{N^{1-\beta-\varepsilon}}} \dd \Gamma(1-\Delta)) \Phi(t) \Big\rangle^{1/2}  \times \Big\langle \Phi_{N,M}(t), \sqrt{\frac{M}{N^{1-\beta-\varepsilon}}} \dd \Gamma(1-\Delta)) \Phi_{N,M}(t) \Big\rangle^{1/2} \nn \\
&\leq C_{t,\varepsilon} \sqrt{\frac{M}{N^{1-3\beta-3\varepsilon}}},
\end{align}
where we used \cref{prop:truncated_stability}, \cref{prop:Bog_stability}.
For the second term, the same computations as in \cref{lem:truncated_dyn} show that
\begin{align}
		\label{eq:Bog_trunc_dyn_3}
\Bigg| \frac{1}{M/2 - 2} \sum_{m=M/2}^{M-3} \langle &\Phi_{N}(t), i [\mathbb{H} (t),\mathds{1}^{\leq m}] \Phi (t) \rangle \Bigg| \leq C_{t,\varepsilon} \frac{N^{\beta+\varepsilon}}{M}
\end{align}
where we have used that $\langle \Phi (t),  \dd \Gamma(1-\Delta)  \Phi (t) \rangle\leq C_{t,\varepsilon} N^{\beta + \varepsilon}$. On the other hand, as in \cref{lem:truncated_dyn} we have
\begin{align}
		\label{eq:Bog_trunc_dyn_4}
 \langle \Phi_{N,M}(0),\mathds{1}^{\leq m} \Phi(0)  \rangle &=\langle  \Phi(0),\mathds{1}^{\leq m} \Phi (0) \rangle \geq 1 - C M^{-1}.
\end{align}

Gathering (\ref{eq:Bog_trunc_dyn_2}),(\ref{eq:Bog_trunc_dyn_3}) and (\ref{eq:Bog_trunc_dyn_4}) we obtain
\begin{equation*}
\Re \langle \Phi_{N,M}(t), \mathds{1}^{\leq m} \Phi(t) \rangle \geq 1 - C_{t,\varepsilon}\left( M^{-1} + \frac{N^{\beta+\varepsilon}}{M} + \sqrt{\frac{M}{N^{1-3\beta-3\varepsilon}}}\right).
\end{equation*}
Together with (\ref{eq:Bog_trunc_dyn_0}) and (\ref{eq:Bog_trunc_dyn_1}), this concludes the proof.
\end{proof}

\subsection{Proof of \cref{theo_1}}
\subsubsection*{Proof of 1)}
The triangle inequality and Lemmas \ref{lem:truncated_dyn} and \ref{lem:bogo_trunc_dyn} give, for all $\varepsilon >0$
\begin{align*}
\|\Phi_N(t) - \Phi(t) \|_{L^2(\mathbb{R}^{3N})} &\leq \|\Phi_N(t) - \Phi_{N,M}(t) \|_{L^2(\mathbb{R}^{3N})} + \|\Phi_{N,M}(t) - \Phi(t) \|_{L^2(\mathbb{R}^{3N})} \\
&\leq C_{t,\varepsilon} \Bigg( \frac{1}{M^{1/2}} + (|e_N| + N)^{1/4}\left( \frac{N^{3(\beta+\varepsilon)/4}}{M}+ \frac{N^{\beta + \varepsilon/2-1/2}}{M^{1/4}}\right)\\
&\qquad +\frac{N^{\beta+\varepsilon}}{M}  + \sqrt{\frac{M}{N^{1-3(\beta+\varepsilon)}}}\Bigg).
\end{align*}
Using that $|e_N| \leq N^{3\beta +1}$ and taking $M = N^{1/2}$ and $\varepsilon $ small enough we obtain
\begin{equation*}
\|\Phi_N(t) - \Phi(t) \|_{L^2(\mathbb{R}^{3N})} \leq C_{t,\varepsilon} \left( N^{-(6\beta - 1)/4 + \varepsilon} + N^{-(7\beta - 2)/4 + \varepsilon}\right),
\end{equation*}
for any $\varepsilon >0$ small enough.

\subsubsection*{Proof of 2)}
Using \cref{lem:truncated_dyn} and \cref{prop:truncated_stability} we have
\begin{align*}
\langle \Phi_N(t), \frac{\mathcal{N}}{N} \Phi_N(t) \rangle 
&= \langle \Phi_N(t), \frac{\mathcal{N}}{N} (\Phi_N(t)-\Phi_{N,M}(t)) \rangle + \langle \Phi_N(t), \frac{\mathcal{N}}{N} \Phi_{N,M}(t) \rangle \\
&\leq \| \Phi_N(t)\|_{L^2(\mathbb{R}^{3N})} \left( \|\Phi_N(t)-\Phi_{N,M}(t)\|_{L^2(\mathbb{R}^{3N})} +  \langle \Phi_{N,M}(t), \frac{\mathcal{N}}{N} \Phi_{N,M}(t) \rangle \right) \\
&\leq C_{t,\varepsilon} \Bigg(\frac{1}{M^{1/2}} + (|e_N| + N)^{1/4}\Bigg( \frac{N^{3(\beta+\varepsilon)/4}}{M}+ \frac{N^{\beta + \varepsilon/2-1/2}}{M^{1/4}}\Bigg)\\
&\qquad \qquad +  N^{-1} \langle \Phi_{N,M}(t), \dd \Gamma(1-\Delta) \Phi_{N,M}(t) \rangle \Bigg) \\
&\leq C_{t,\varepsilon} \left(\frac{1}{M^{1/2}} + \Bigg( \frac{N^{3\beta/2+3\varepsilon/4 + 1/4}}{M}+ \frac{N^{7\beta/4 + \varepsilon/2-1/4}}{M^{1/4}}\Bigg)+  N^{\beta-1}\right).
\end{align*}
Taking $M=N^{\alpha}$ with $\beta < \alpha < 1 - \beta $ we obtain after optimizing over $\alpha$
\begin{align*}
\|\Gamma^{(1)}_{\Psi_N(t)} - \ket{u_N(t)} \bra{u_N(t)}\|_{\mathfrak{S}_1} 
&\leq \langle \Phi_N(t), \frac{\mathcal{N}}{N} \Phi_N(t) \rangle  \\
&\leq C_{t,\varepsilon} \left(N^{-(3 -10\beta )/4 + \varepsilon} + N^{-(1-4\beta )/4 + \varepsilon} \right),
\end{align*}
for any $\varepsilon >0$ small enough.

\subsubsection*{Proof of 3)}
Using that when $\widehat{w} \geq 0$ we have the bound $|e_N| \leq C\left( N + N^{3\beta} \right)$, the same computation as before with $M = N^{-(1-\beta)/2}$ and assuming $1/6 < \beta < 1/5$ shows that
\begin{align*}
\|\Phi_N(t) - \Phi(t) \|_{L^2(\mathbb{R}^{3N})} \leq C_{t,\varepsilon} N^{-(1 - 5 \beta )/4 + \varepsilon},
\end{align*}
for any $\varepsilon >0$ small enough. Now take $M=N^{\alpha}$ with $\beta < \alpha < 1-\beta$, with the same computation as before, after optimizing over $\alpha$, we obtain
\begin{align*}
\|\Gamma^{(1)}_{\Psi_N(t)} - \ket{u_N(t)} \bra{u_N(t)}\|_{\mathfrak{S}_1} 
&\leq \langle \Phi_N(t), \frac{\mathcal{N}}{N} \Phi_N(t) \rangle  \\
&\leq C_{t,\varepsilon} \left(N^{-(1-\beta)/2 + \varepsilon} + N^{-(2-5\beta)/2 + \varepsilon} + N^{-(3 - 8\beta)/4 +\varepsilon}\right),
\end{align*}
for any $\varepsilon >0$. \qed

%%%%%%%%%%%%%%%%%%%%%%%%%%%%%%%%%%%%%%%%%%%%%%%%%%%%%%%%%%%%%%%%%%%%%%%%%%%%%%
%%%%%%%%%%%%%%%%%%%%%%%%%%%%%%%%%%%%%%%%%%%%%%%%%%%%%%%%%%%%%%%%%%%%%%%%%%%%%%
%\input{H2_regularity}
%%%%%%%%%%%%%%%%%%%%%%%%%%%%%%%%%%%%%%%%%%%%%%%%%%%%%%%%%%%%%%%%%%%%%%%%%%%%%%
%%%%%%%%%%%%%%%%%%%%%%%%%%%%%%%%%%%%%%%%%%%%%%%%%%%%%%%%%%%%%%%%%%%%%%%%%%%%%%

\bibliographystyle{siam} %style de biblio
\bibliography{biblio} %nom du fichier biblio.bib (dans le mÍme dossier que le fichier latex qu'on compile)

\begin{thebibliography}{10}

\bibitem{Aikawa-12}
{\sc K.~Aikawa, A.~Frisch, M.~Mark, S.~Baier, A.~Rietzler, R.~Grimm, and
  F.~Ferlaino}, {\em {B}ose-{E}instein condensation of {E}rbium}, Physical
  review letters, 108 (2012), p.~210401.

\bibitem{CorWie-95}
{\sc M.~H. Anderson, J.~R. Ensher, M.~R. Matthews, C.~E. Wieman, and E.~A.
  Cornell}, {\em Observation of {B}ose-{E}instein condensation in a dilute
  atomic vapor}, Science, 269 (1995), pp.~198--201.

\bibitem{BahGer-99}
{\sc H.~Bahouri and P.~G{\'e}rard}, {\em High frequency approximation of
  solutions to critical nonlinear wave equations}, Amer. J. Math., 121 (1999),
  pp.~131--175.

\bibitem{Baillie-16}
{\sc D.~Baillie, R.~Wilson, R.~Bisset, and P.~Blakie}, {\em Self-bound dipolar
  droplet: A localized matter wave in free space}, Physical Review A, 94
  (2016), p.~021602.

\bibitem{Beaufils-08}
{\sc Q.~Beaufils, R.~Chicireanu, T.~Zanon, B.~Laburthe-Tolra, E.~Mar\'echal,
  L.~Vernac, J.-C. Keller, and O.~Gorceix}, {\em All-optical production of
  chromium {B}ose-{E}instein condensates}, Phys. Rev. A, 77 (2008), p.~061601.

\bibitem{BenOliSch-15}
{\sc N.~Benedikter, G.~de~Oliveira, and B.~Schlein}, {\em Quantitative
  derivation of the {G}ross-{P}itaevskii equation}, Comm. Pure Appl. Math., 68
  (2015), pp.~1399--1482.

\bibitem{Bose-24}
{\sc S.~N. {B}ose}, {\em {Plancks Gesetz und Lichtquantenhypothese}}, Z. Phys.,
  26 (1924), pp.~178--181.

\bibitem{CarMarkSpa-08}
{\sc R.~Carles, P.~A. Markowich, and C.~Sparber}, {\em On the
  {G}ross-{P}itaevskii equation for trapped dipolar quantum gases},
  Nonlinearity, 21 (2008), pp.~2569--2590.

\bibitem{Cazenave}
{\sc T.~Cazenave}, {\em Semilinear {S}chr\"odinger equations}, vol.~10 of
  Courant Lecture Notes in Mathematics, New York University Courant Institute
  of Mathematical Sciences, New York, 2003.

\bibitem{CheHol-16}
{\sc X.~Chen and J.~Holmer}, {\em Focusing quantum many-body dynamics: the
  rigorous derivation of the 1{D} focusing cubic nonlinear {S}chr\"{o}dinger
  equation}, Arch. Ration. Mech. Anal., 221 (2016), pp.~631--676.

\bibitem{CheHol-16b}
\leavevmode\vrule height 2pt depth -1.6pt width 23pt, {\em The rigorous
  derivation of the 2{D} cubic focusing {NLS} from quantum many-body
  evolution}, Int. Math. Res. Not. IMRN,  (2017), pp.~4173--4216.

\bibitem{Chomaz-16}
{\sc L.~Chomaz, S.~Baier, D.~Petter, M.~Mark, F.~W{\"a}chtler, L.~Santos, and
  F.~Ferlaino}, {\em Quantum-fluctuation-driven crossover from a dilute
  {B}ose-{E}instein condensate to a macrodroplet in a dipolar quantum fluid},
  Physical Review X, 6 (2016), p.~041039.

\bibitem{Chong-16}
{\sc J.~J.~W. Chong}, {\em Dynamics of large boson systems with attractive
  interaction and a derivation of the cubic focusing {NLS} in {$\mathbb R\sp
  3$}}, arXiv preprint arXiv:1608.01615,  (2016).

\bibitem{Duo01}
{\sc J.~Duoandikoetxea}, {\em Fourier analysis, volume 29 of graduate studies
  in mathematics}, {A}merican {M}athematical {S}ociety, {P}rovidence, {RI},
  (2001).

\bibitem{Ein-24}
{\sc A.~Einstein}, {\em Quantentheorie des einatomigen idealen Gases}, Akademie
  der Wissenshaften, in Kommission bei W. de Gruyter, 1924.

\bibitem{ErdSchYau-06}
{\sc L.~Erd\H{o}s, B.~Schlein, and H.-T. Yau}, {\em Derivation of the
  {G}ross-{P}itaevskii hierarchy for the dynamics of {B}ose-{E}instein
  condensate}, Comm. Pure Appl. Math., 59 (2006), pp.~1659--1741.

\bibitem{ErdSchYau-07}
\leavevmode\vrule height 2pt depth -1.6pt width 23pt, {\em Derivation of the
  cubic non-linear {S}chr\"{o}dinger equation from quantum dynamics of
  many-body systems}, Invent. Math., 167 (2007), pp.~515--614.

\bibitem{ErdSchYau-09}
\leavevmode\vrule height 2pt depth -1.6pt width 23pt, {\em Rigorous derivation
  of the {G}ross-{P}itaevskii equation with a large interaction potential}, J.
  Amer. Math. Soc., 22 (2009), pp.~1099--1156.

\bibitem{ErdSchYau-10}
\leavevmode\vrule height 2pt depth -1.6pt width 23pt, {\em Derivation of the
  {G}ross-{P}itaevskii equation for the dynamics of {B}ose-{E}instein
  condensate}, Ann. of Math. (2), 172 (2010), pp.~291--370.

\bibitem{GriWerHenStuPfa-05}
{\sc A.~{Griesmaier}, J.~{Werner}, S.~{Hensler}, J.~{Stuhler}, and T.~{Pfau}},
  {\em {Bose-Einstein Condensation of Chromium}}, Phys. {R}ev. {L}ett., 94
  (2005), p.~160401.

\bibitem{HenNatPoh-10}
{\sc N.~Henkel, R.~Nath, and T.~Pohl}, {\em Three-dimensional roton excitations
  and supersolid formation in rydberg-excited bose-einstein condensates},
  Physical review letters, 104 (2010), p.~195302.

\bibitem{JebPic-18b}
{\sc M.~Jeblick and P.~Pickl}, {\em {Derivation of the time dependent
  Gross-Pitaevskii equation for a class of non purely positive potentials}},
  arXiv preprint arXiv:1801.04799,  (2018).

\bibitem{JebPic-18}
\leavevmode\vrule height 2pt depth -1.6pt width 23pt, {\em Derivation of the
  {T}ime {D}ependent {T}wo {D}imensional {F}ocusing {NLS} {E}quation}, J. Stat.
  Phys., 172 (2018), pp.~1398--1426.

\bibitem{LahMenSanLewPfau-09}
{\sc T.~Lahaye, C.~Menotti, L.~Santos, M.~Lewenstein, and T.~Pfau}, {\em The
  physics of dipolar bosonic quantum gases}, Rep. Prog. Phys., 72 (2009),
  p.~126401.

\bibitem{LewNamRou-15}
{\sc M.~Lewin, P.~T. Nam, and N.~Rougerie}, {\em The mean-field approximation
  and the non-linear {S}chr\"{o}dinger functional for trapped {B}ose gases},
  Trans. Amer. Math. Soc., 368 (2016), pp.~6131--6157.

\bibitem{LewNamSch-15}
{\sc M.~Lewin, P.~T. Nam, and B.~Schlein}, {\em Fluctuations around {H}artree
  states in the mean-field regime}, Amer. J. Math., 137 (2015), pp.~1613--1650.

\bibitem{LewNamSerSol-13}
{\sc M.~Lewin, P.~T. Nam, S.~Serfaty, and J.~P. Solovej}, {\em Bogoliubov
  spectrum of interacting {B}ose gases}, Comm. Pure Appl. Math., in press
  (2013).

\bibitem{LewNamRou-17b}
{\sc M.~{Lewin}, P.~{Th{\`a}nh Nam}, and N.~{Rougerie}}, {\em {A note on 2D
  focusing many-boson systems}}, Proc. Amer. Math. Soc., 145 (2017),
  pp.~2441--2454.

\bibitem{LieSol-01}
{\sc E.~H. Lieb and J.~P. Solovej}, {\em Ground state energy of the
  one-component charged {B}ose gas}, Commun. Math. Phys., 217 (2001),
  pp.~127--163.

\bibitem{LuBurYouLev-11}
{\sc M.~Lu, N.~Q. Burdick, S.~H. Youn, and B.~L. Lev}, {\em {Strongly Dipolar
  Bose-Einstein Condensate of Dysprosium}}, Phys. Rev. Lett., 107 (2011),
  p.~190401.

\bibitem{NamNap-17d}
{\sc P.~T. Nam and M.~Napi{\'o}rkowski}, {\em Norm approximation for many-body
  quantum dynamics: focusing case in low dimensions}, arXiv preprint
  arXiv:1710.09684,  (2017).

\bibitem{NamNapSol-16}
{\sc P.~T. Nam, M.~Napi\'{o}rkowski, and J.~P. Solovej}, {\em Diagonalization
  of bosonic quadratic {H}amiltonians by {B}ogoliubov transformations}, J.
  Funct. Anal., 270 (2016), pp.~4340--4368.

\bibitem{Pickl-10}
{\sc P.~Pickl}, {\em Derivation of the time dependent {G}ross-{P}itaevskii
  equation without positivity condition on the interaction}, J. Stat. Phys.,
  140 (2010), pp.~76--89.

\bibitem{Pickl-11}
{\sc P.~Pickl}, {\em A simple derivation of mean-field limits for quantum
  systems}, Lett. Math. Phys., 97 (2011), pp.~151--164.

\bibitem{SanShlLew-03}
{\sc L.~Santos, G.~Shlyapnikov, and M.~Lewenstein}, {\em Roton-maxon spectrum
  and stability of trapped dipolar {B}ose-{E}instein condensates}, Phys. {R}ev.
  {L}ett., 90 (2003), p.~250403.

\bibitem{stein1970singular}
{\sc E.~M. Stein}, {\em Singular integrals and differentiability properties of
  functions}, vol.~2, Princeton university press, 1970.

\bibitem{Tri-18}
{\sc A.~Triay}, {\em Derivation of the dipolar {G}ross-{P}itaevskii energy},
  SIAM J. Math. Anal., 50 (2018), pp.~33--63.

\end{thebibliography}
\end{document}